\documentclass[a4paper, english, autoref, thm-restate, nolineno]{socg-lipics-v2021}

\usepackage[utf8]{inputenc} 	
\usepackage[noabbrev,capitalise]{cleveref}
\usepackage{microtype}
\usepackage{enumitem}
\usepackage{textcomp}
\usepackage{cite}
\usepackage{xspace}

\newcommand{\e}{\ensuremath{\varepsilon}\xspace}
\newcommand{\Real}{\mathbb{R}}
 \newcommand{\tangent}{\ensuremath{T}\xspace}
\newcommand{\normal}{\ensuremath{N}\xspace}

\newcommand{\cham}[1]{\ensuremath{\bar{#1}}\xspace}

\def\EMPH#1{\textbf{\emph{\boldmath #1}}}
\def\into{\DOTSB\hookrightarrow}
\def\set#1{\{ #1 \}}
\def\abs#1{\mathopen| #1 \mathclose|}		

\title{Chasing Puppies:\\ Mobile Beacon Routing on Closed Curves}

\titlerunning{Chasing Puppies: Mobile Beacon Routing on Closed Curves}

\author{Mikkel Abrahamsen}{Basic Algorithms Research Copenhagen (BARC), University of Copenhagen, Denmark}{miab@di.ku.dk}{http://orcid.org/0000-0003-2734-4690}{Partially supported by the VILLUM Foundation grant 16582.}
\author{Jeff Erickson}{University of Illinois at Urbana-Champaign, USA}{jeffe@illinois.edu}{https://orcid.org/0000-0002-5253-2282}{}
\author{Irina Kostitsyna}{Eindhoven University of Technology, Netherlands}{i.kostitsyna@tue.nl}{}{}
\author{Maarten Löffler}{Utrecht University, Netherlands}{m.loffler@uu.nl}{}{Partially supported by the Dutch Research Council (NWO) under project number 614.001.504 and 628.011.005.}
\author{Tillmann Miltzow}{Utrecht University, Netherlands}{t.miltzow@uu.nl}{https://orcid.org/0000-0003-4563-2864}{Supported by the Dutch Research Council (NWO) under Veni grant EAGER}
\author{Jérôme Urhausen}{Utrecht University, Netherlands}{j.e.urhausen@uu.nl}{}{Supported by the Dutch Research Council (NWO); 612.001.651.}
\author{Jordi Vermeulen}{Utrecht University, Netherlands}{j.l.vermeulen@uu.nl}{}{Supported by the Dutch Research Council (NWO); 612.001.651.}
\author{Giovanni Viglietta}{Japan Advanced Institute of Science and Technology, Nomi City, Ishikawa, Japan}{johnny@jaist.ac.jp}{https://orcid.org/0000-0001-6145-4602}{}

\authorrunning{Abrahamsen, Erickson, Kostitsyna, Löffler, Miltzow, Urhausen, Vermeulen, Viglietta} 

\Copyright{Jeff Erickson, Irina Kostitsyna, Maarten Löffler, Tillmann Miltzow, Jérôme Urhausen, Jordi Vermeulen, Giovanni Viglietta} 

\ccsdesc[500]{Theory of computation~Computational geometry} 

\keywords{Beacon routing, navigation, generic smooth curves, puppies} 

\category{} 

\supplement{https://github.com/viglietta/Chasing-Puppies}

\acknowledgements{Thanks to Ivor van der Hoog, Marc van Kreveld, and Frank Staals for helpful discussions in the early stages of this work, and to Joseph O'Rourke for clarifying the history of the problem.  Portions of this work were done while the second author was visiting Utrecht University.}

\nolinenumbers 

\hideLIPIcs  

\begin{document}

\maketitle

\begin{abstract}
We solve an open problem posed by Michael Biro at CCCG 2013 that was inspired by his and others’ work on beacon-based routing.
Consider a human and a puppy on a simple closed curve in the plane.
The human can walk along the curve at bounded speed and change direction as desired.
The puppy runs with unbounded speed along the curve as long as the Euclidean straight-line distance to the human is decreasing, so that it is always at a point on the curve where the distance is locally minimal.
Assuming that the curve is smooth (with some mild genericity constraints) or a simple polygon, we prove that the human can always catch the puppy in finite time.
\end{abstract}

\clearpage
\section{Introduction}
\label{S:intro}

You have lost your puppy somewhere on a simple closed curve.
Both of you are forced to stay on the curve.
You can see each other and both want to reunite.
The problem is that the puppy runs infinitely faster than you, and it believes naively that it is always a good idea to minimize its straight-line distance to you.
What do you do?

To be more precise, let $\gamma\colon S^1\into \Real^2$ be a simple closed curve in the plane, which we informally call the \emph{track}.
Two special points move around the track, called the \emph{puppy} $p$ and the \emph{human} $h$.
The human can walk along the track at bounded speed and change direction as desired.
The puppy runs with unbounded speed along the track as long as its Euclidean straight-line distance to the human is decreasing, until it reaches a point on the curve where the distance is locally minimized.  As the human moves along the track, the puppy moves to stay  at a local distance minimum.  The human's goal is to move in such a way that the puppy and the human meet.  See \cref{F:init-example} for a simple example.

\begin{figure}[ht]
\centering
\includegraphics[scale=0.48]{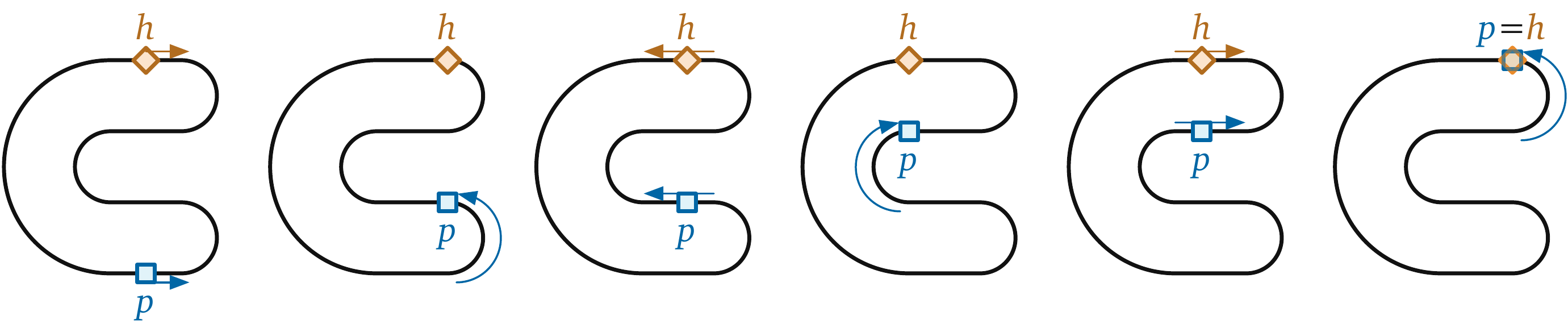}
\caption{Catching the puppy.}
\label{F:init-example}
\end{figure}

In this paper we show that it is always possible to reunite with the puppy, under the assumption that the curve is well-behaved in a sense to be defined.

The problem was posed in a different guise at the open problem session of the 25th Canadian Conference on Computational Geometry (CCCG 2013) by Michael Biro.  In Biro's formulation, 
the track was a railway, the human a locomotive, and the puppy a train carriage that was attracted to an infinitely strong magnet installed in the locomotive.

\begin{figure}[ht]
\centering
\includegraphics[page=4]{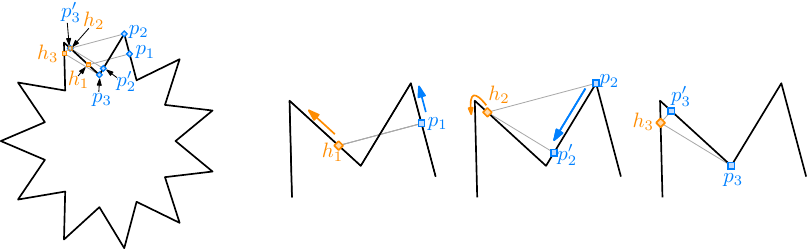}
\caption{If the human keeps walking counterclockwise from $h_1$, the human and the puppy will never meet.
To the right are closeups of two of the spikes of the star.
}
\label{F:intro2}
\end{figure}

Returning to our formulation of catching a puppy, it was also asked if the human will always catch the puppy by choosing an arbitrary direction and walking only in that direction.
This turns out not to be the case; consider the star-shaped track in \cref{F:intro2}.  Suppose the human and puppy start at points $h_1$ and $p_1$, respectively, and the human walks counterclockwise around the track.
When the human reaches $h_2$, the puppy runs from $p_2$ to~$p'_2$.
When the human reaches $h_3$, the puppy runs from $p_3$ to $p'_3$.
Then the pattern repeats indefinitely. 
Examples of this type, where the human walking in the wrong direction will never catch the puppy, were independently discovered during the conference by some of the authors and by David Eppstein.

\subsection{Related work}

Biro’s problem was inspired by his and others’ work on \emph{beacon-based geometric routing}, a generalization of both greedy geometric routing and the art gallery problem introduced at the 2011 Fall Workshop on Computational Geometry \cite{bgikm-bre-11} and the 2012 Young Researchers Forum \cite{bgikm-bbspd-12}, and further developed in Biro’s PhD thesis~\cite{biro2013beacon} and papers~\cite{DBLP:conf/wads/BiroIKM13, DBLP:conf/cccg/BiroGIKM13}.  A \emph{beacon} is a stationary point object that can be activated to create a “magnetic pull” towards itself everywhere in a given polygonal domain $P$.  When a beacon at point $b$ is \emph{activated}, a point object $p$ moves moves greedily to decrease its Euclidean distance to $b$, alternately moving through the interior of $P$ and sliding along its boundary, until it either reaches $b$ or gets stuck at a “dead point” where Euclidean distance is minimized.  By activating different beacons one at a time, one can route a moving point object through the domain.
%
%
Initial results for this model by Biro and his colleagues \cite{bgikm-bre-11, bgikm-bbspd-12, biro2013beacon, DBLP:conf/wads/BiroIKM13, DBLP:conf/cccg/BiroGIKM13} sparked significant interest and subsequent work in the community~\cite{DBLP:journals/comgeo/BaeSV19,kostitsyna_et_al:LIPIcs:2018:8768,DBLP:conf/cccg/KouhestaniRS15,DBLP:conf/cccg/KouhestaniRS16,DBLP:journals/comgeo/KouhestaniRS18,DBLP:conf/cccg/Shermer15, aacns-bcop-17, aacns-tbico-20, cm-cbrtd-20}.
More recent works have also studied how to utilize objects that repel points instead of attracting them~\cite{DBLP:journals/comgeo/BoseS20,DBLP:conf/cocoa/MozafariS18}.

Biro’s problem can also be viewed as a novel variant of classical \emph{pursuit} problems, which have been an object of intense study for centuries~\cite{n-cempe-07}.  The oldest pursuit problems ask for a description of the \emph{pursuit curve} traced by a \emph{pursuer} moving at constant speed directly toward a \emph{target} moving along some other curve.  Pursuit curves were first systematically studied by Bouguer \cite{b-ncapd-32} and de Maupertuis \cite{m-cp-32} in 1732, who used the metaphor of a pirate overtaking a merchant ship; another notable example is Hathaway’s problem \cite{h-psp2-20}, which asks for the the pursuit curve of a dog swimming at unit speed in a circular lake directly toward a duck swimming at unit speed around its circumference.  In more modern \emph{pursuit-evasion} problems, starting with Rado’s famous “lion and man” problem \cite[pp.114--117]{b-lm-86}, the pursuer and target both move strategically within some geometric domain; the pursuer attempts to \emph{capture} the target by making their positions coincide while the target attempts to evade capture.  Countless variants of pursuit-evasion problems have been studied, with multiple pursuers and/or targets, different classes of domains, various constraints on motion or visibility, different capture conditions, and so on.  Biro’s problem can be naturally described as a \emph{cooperative pursuit} or \emph{pursuit-attraction} problem, in which a strategic target (the human) \emph{wants} to be captured by a greedy pursuer (the puppy).

Kouhestani and Rappaport~\cite{kouhestani2017} studied a natural variant of Biro’s problem, which we can recast as follows.
A \emph{guppy} is restricted to a closed and simply-connected \emph{lake}, while the human is restricted to the boundary of the lake.  The guppy swims with unbounded speed to decrease its Euclidean distance to the human as quickly as possible.  Kouhestani and Rappaport described a polynomial-time algorithm that finds a strategy for the human to catch the guppy, if such a strategy exists, given a simple polygon as input; they also conjectured that a capturing strategy always exists. 
Abel, Akitaya, Demaine, Demaine, Hesterberg, Korman, Ku, and Lynch~\cite{DBLP:journals/corr/abs-2006-01202} recently proved that for some polygons and starting configurations, the human cannot catch the guppy, even if the human is allowed to walk in the exterior of the polygon, thereby disproving Kouhestani and Rappaport’s conjecture.  Their simplest counterexample is an orthogonal polygon with about 50 vertices.

\subsection{Our results}

Before describing our results in detail, we need to carefully define the terms of the problem.  The \emph{track} is a simple closed curve $\gamma\colon S^1\into \Real^2$.  We consider the motion of two points on this curve, called the \emph{human} (or \emph{beacon} or \emph{target}) and the \emph{puppy} (or \emph{pursuer}).  A \emph{configuration} is a pair $(x, y) \in S^1\times S^1$ that specifies the locations $h = \gamma(x)$ and $p = \gamma(y)$ for the human and puppy, respectively.  Let $D(x, y)$ denote the straight-line Euclidean distance between these two points.  When the human is located at $h = \gamma(x)$, the puppy moves from $p = \gamma(y)$ to greedily decrease its distance to the human, as follows.
\begin{itemize}
\item
If $D(x, y+\e) < D(x, y)$ for all sufficiently small $\e>0$, the puppy runs forward along the track, by increasing the parameter $y$.
\item 
If $D(x, y-\e) < D(x, y)$ for all sufficiently small $\e>0$, the puppy runs backward along the track, by decreasing the parameter $y$.
\end{itemize}
If both of these conditions hold, the puppy runs in an arbitrary direction.  While the puppy is running, the human remains stationary.  If neither condition holds, the configuration is \emph{stable}; the puppy does not move until the human does.  When the configuration is stable, the human can walk in either direction along the track; the puppy walks along the track in response to keep the configuration stable, until it is forced to run again.  The human's goal is to \emph{catch} the puppy; that is, to reach a configuration in which the two points coincide. 

Our main result is that the human can always catch the puppy in finite time, starting from any initial configuration, provided the track is either a generic simple smooth curve or an arbitrary simple polygon.

The remainder of the paper is structured as follows.  We begin in \cref{S:warmup} by  considering some variants and special cases of the problem.  In particular, we give a simple self-contained proof of our main result for the special case of orthogonal polygons.

We consider generic smooth tracks in \cref{S:smooth-diagrams,S:dexterity}. Specifically, in \cref{S:smooth-diagrams} we define two important diagrams, which we call the \emph{attraction diagram} and the \emph{dual attraction diagram}, and prove some useful structural results.  At a high level, the attraction diagram is a decomposition of the configuration space $S^1\times S^1$ according to the puppy's behavior, similar to the \emph{free space diagrams} introduced by Alt and Godau to compute Fréchet distance \cite{ag-cfdbt-95}. We show that for a sufficiently generic smooth track, the attraction diagram consists of a finite number of disjoint simple closed \emph{critical} curves, exactly two of which are topologically nontrivial.  Then in \cref{S:dexterity}, we argue that the human can catch the puppy on any track whose attraction diagram has this structure.

In \cref{S:polygons}, we describe an extension of our analysis from smooth curves to simple polygonal tracks.  Because polygons do not have well-defined tangent directions at their vertices, this extension requires explicitly modeling the puppy's direction of motion in addition to its location.  We first prove that the human can catch the puppy on a polygon that has no acute vertex angles and where no three vertices form a right angle; under these conditions, the attraction diagram has exactly the same structure as for generic smooth curves.  We then reduce the problem for arbitrary simple polygons to this special case by \emph{chamfering}---cutting off a small triangle at each vertex---and arguing that any strategy for catching the puppy on the chamfered track can be pulled back to the original polygon.  

Finally, we close the paper by suggesting several directions for further research.

\section{Warmup: other settings and a special case}
\label{S:warmup}

In this section, we discuss two variants of Biro’s problem and the special case of orthogonal polygons.

In the first variant, both the human $h$ and the puppy $p$ are allowed to move anywhere in the interior and on the boundary of a simple polygon $P$.  Here, as in beacon routing and Kouhestani and Rappaport’s variant~\cite{kouhestani2017, DBLP:journals/corr/abs-2006-01202}, the puppy moves greedily to decrease its Euclidean distance to the human, alternately moving through the interior of $P$ and sliding along its boundary.

As we will show in \cref{thm:pol}, $h$ has a simple strategy to catch $p$ in this setting, essentially by walking along the dual graph of any triangulation.  This is an interesting contrast to the proof by Abel et al. \cite{DBLP:journals/corr/abs-2006-01202} that $h$ and $p$ cannot always meet when $h$ is restricted to the \emph{exterior} of $P$ and $p$ to the interior.  Our main result that $h$ and $p$ can meet when both are restricted to the \emph{boundary} of $P$ (even for a much wider class of simple closed curves), somehow sits in between these other two variants.

\begin{figure}[ht]
\centering
\includegraphics[page=3]{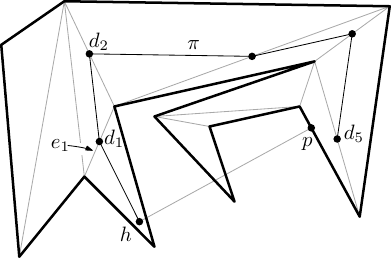}
\caption{The proposed strategy when $h$ and $p$ are restricted to the interior of a simple polygon~$P$.
The human $h$ will follow the path $\pi$.
Note that the triangle containing $p$ will change before $h$ reaches $d_1$, and $\pi$ will be updated accordingly.}
\label{F:thm:pol}
\end{figure}

When both $h$ and $p$ are restricted to the interior of $P$, we propose the following strategy for $h$; see \cref{F:thm:pol}.
Let $\mathcal T$ be a triangulation of $P$ and let $t_1,\ldots,t_k$ be the path of pairwise adjacent triangles in $\mathcal T$ such that $h\in t_1$ and $p\in t_k$.
Let $e_i$ be the common edge of $t_i$ and $t_{i+1}$ and let $d_i$ be the midpoint of $e_i$.
Let $\pi=hd_1d_2\ldots d_{k-1}$ be a path from $h$ to $d_{k-1}$, which is contained in the triangles $t_1,\ldots,t_{k-1}$.
The human starts walking along $\pi$.
As soon as the puppy enters a new triangle, the human recomputes $\pi$ as described and follows the new path.

\begin{theorem}\label{thm:pol}
The proposed strategy will make $h$ and $p$ meet.
\end{theorem}

\begin{proof}
First, we observe that if the puppy ever enters the triangle $t_1$ that
is occupied by the human, then the puppy and the human will meet
immediately.
Assume that the human does not meet the puppy right from the beginning.
The region $P\setminus t_1$ consists of one, two, or three polygons, one of which $P_p$ contains $p$.
Thus, whenever the human moves from one triangle to another,
the set of triangles that can possibly contain $p$ shrinks.
We conclude that the human and the puppy must meet eventually.
\end{proof}

In our second variant, the human and the puppy are both restricted to a simple, closed curve~$\gamma$ in $\Real^3$.  Here it is easy to construct curves on which $h$ and $p$ will never meet; the simplest example is a “double loop” that approximately winds twice around a planar circle, as shown in \cref{F:doubleLoop}. 

\begin{figure}[ht]
\centering
\includegraphics[page=2]{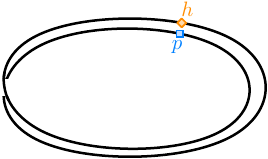}
\caption{A double loop in $\Real^3$; the human and puppy will never meet.}
\label{F:doubleLoop}
\end{figure}

Finally, we consider the special case of Biro’s original problem where the track $\gamma$ is the boundary of an orthogonal polygon in the plane.  This special case of our main results admits a much simpler self-contained proof.

\begin{theorem}
\label{Th:ortho}
The human can catch the puppy on any simple orthogonal polygon, by walking counterclockwise around the polygon at most twice.
\end{theorem}

\begin{proof}
Let~$P$ be an arbitrary simple orthogonal polygon.
Let~$u_1$ be its leftmost point with the maximum $y$-coordinate, and $u_2$ be the next boundary vertex of $P$ in the clockwise order (see \cref{F:ortho}).
Finally, let~$\ell$ be the horizontal line supporting the segment $u_1u_2$.

We break the motion of the human into two phases.
In the first phase, the human moves counterclockwise around $P$ from their starting location to $u_1$.
If the human catches the puppy during this phase, we are done, so assume otherwise.
In the second phase, the human walks counterclockwise around $P$ starting from~$u_1$ to $u_2$.

We claim that the puppy $p$ is never in the interior of the segment $u_1u_2$ during the second phase; thus, $p$ always lies on the closed counterclockwise subpath of~$P$ from $h$ to  $u_2$ (or less formally, ``between $h$ and $u_2$'').  This claim implies that the human and the puppy are united during the second phase on $u_2$ at the latest.

\begin{figure}[ht]
\centering
\includegraphics[scale=0.5]{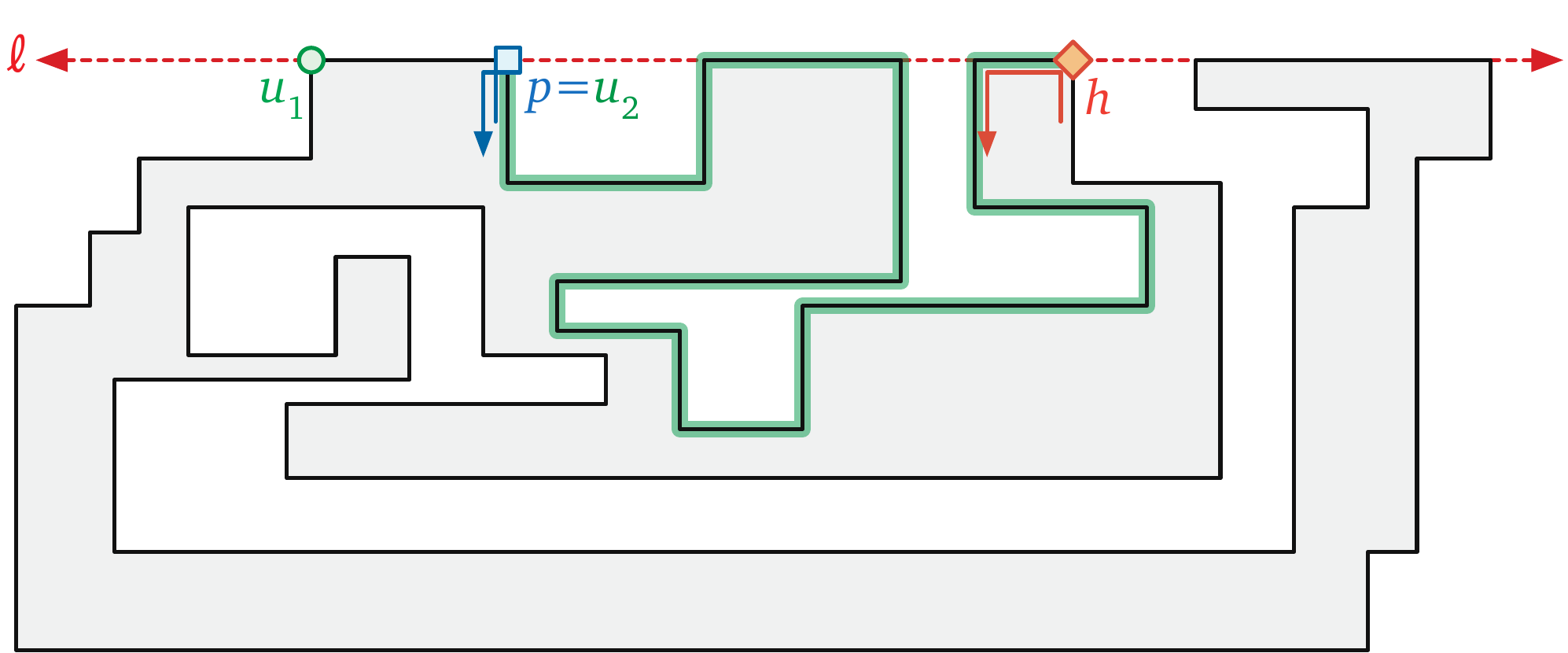}
\caption{Proof of \cref{Th:ortho}.  During the human's second trip around $P$, the puppy lies between $u_2$ and the human.}
\label{F:ortho}
\end{figure}

The puppy must first be at $u_2$ if it ever wants to be in
the interior of $u_1u_2$.
So consider any moment during the second phase when~$p$ 
moves upward to the vertex $u_2$.
At that moment,~$h$ must be on the line~$\ell$ to the right of~$p$.
(For any point~$a$ below~$\ell$,  there is a point~$b$ below~$u_2$ that is closer to $a$ than~$u_2$.)
Thus, the puppy will stay on $u_2$ as long as $h$ is on $\ell$.
As soon as~$h$ leaves~$\ell$ the puppy will leave $u_2$ downward.
Thus the puppy can never go to the interior of the edge $u_1u_2$.
\end{proof}

The star-shaped track in \cref{F:intro2} shows that this simple argument does not extend to arbitrary polygons, even with a constant number of edge directions.

\section{Diagrams of smooth tracks}
\label{S:smooth-diagrams}

We first formalize both the problem and our solution under the assumption that the track is a generic smooth simple closed curve $\gamma\colon S^1\into \Real^2$.  In particular, for ease of exposition, we assume that $\gamma$ is regular and $C^3$, meaning it has well-defined continuous first, second, and third derivatives, and its first derivative is nowhere zero.  We also assume $\gamma$ satisfies some additional genericity constraints, to be specified later.  We consider polygonal tracks in \cref{S:polygons}.
 
\subsection{Configurations and genericity assumptions}

We analyze the behavior of the puppy in terms of the \emph{configuration space} $S^1\times S^1$, which is the standard torus.  Each configuration point $(x, y) \in S^1\times S^1$ corresponds to the human being located at $h = \gamma(x)$ and the puppy being located at $p = \gamma(y)$.

For any configuration $(x,y)$, recall that $D(x,y)$ denotes the straight-line Euclidean distance between the points $\gamma(x)$ and $\gamma(y)$.
We classify all configurations $(x, y) \in S^1\times S^1$ into three types, according to the sign of the partial derivative of distance with respect to the puppy's position.
\begin{itemize}\itemsep0pt
\item
$(x,y)$ is a \emph{forward} configuration if $\frac{\partial}{\partial y} D(x,y) < 0$.
\item
$(x,y)$ is a \emph{backward} configuration if $\frac{\partial}{\partial y} D(x,y) > 0$.
\item
$(x,y)$ is a \emph{critical} configuration if $\frac{\partial}{\partial y} D(x,y) = 0$.
\end{itemize}
Starting in any forward (resp.~backward) configuration, the puppy automatically runs forward (resp.~backward) along the track $\gamma$. 
Genericity implies that there are a finite number of critical configurations $(x, y)$ with any fixed value of $x$, or with any fixed value of $y$.  We further classify the critical configurations as follows:
\begin{itemize}\itemsep0pt
\item
$(x,y)$ is a \emph{stable} critical configuration if $\frac{\partial^2}{\partial y^2} D(x,y) > 0$.
\item
$(x,y)$ is an \emph{unstable} critical configuration if $\frac{\partial^2}{\partial y^2} D(x,y) < 0$.
\item
$(x,y)$ is a \emph{forward pivot} configuration if $\frac{\partial^2}{\partial y^2} D(x,y) = 0$ and $\frac{\partial^3}{\partial y^3} D(x,y) < 0$.
\item
$(x,y)$ is a \emph{backward pivot} configuration if $\frac{\partial^2}{\partial y^2} D(x,y) = 0$ and $\frac{\partial^3}{\partial y^3} D(x,y) > 0$.
\end{itemize}
In any stable configuration, the puppy's distance to the human is locally minimized, so the puppy does not move unless the human moves.  
In any unstable configuration, the puppy can decrease its distance by running in either direction.  
Finally, in any forward (resp.~backward) pivot configuration, the puppy can decrease its distance by moving 
in one direction but not the other, and thus automatically runs forward (resp.~backward) along the track.

Critical points can also be characterized geometrically as follows. Refer to \cref{F:critical}.  A configuration $(x,y)$ is critical if the human $\gamma(x)$ lies on the line $\normal(y)$ normal to $\gamma$ at the puppy's location $\gamma(y)$.  Let $C(y)$ denote the center of curvature of the track at $\gamma(y)$.  Then $(x, y)$ is a pivot configuration if $\gamma(x) = C(y)$, a stable critical configuration if the open ray from $C(y)$ through the human point $\gamma(x)$ contains the puppy point $\gamma(y)$, and an unstable critical configuration otherwise.

\begin{figure}[ht]
\centering
\includegraphics[scale=0.5]{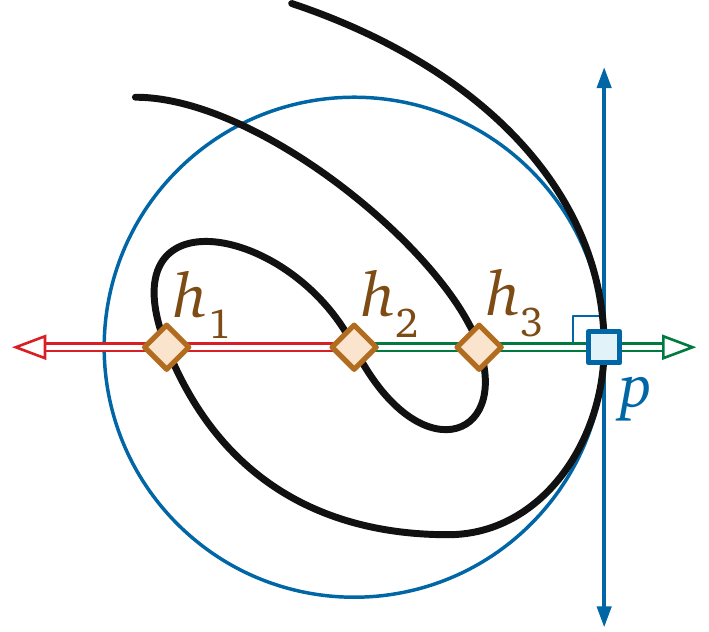}
\caption{Three critical configurations: $(h_1, p)$ is unstable; $(h_2, p)$ is a pivot configuration, and $(h_3, p)$ is stable.}
\label{F:critical}
\end{figure}

Genericity of the track $\gamma$ implies that this classification of critical configurations is exhaustive, and moreover, that the set of pivot configurations is finite.  
In particular, our analysis requires that in any pivot configuration $(x,y)$, the puppy point $\gamma(y)$ is not a local curvature minimum or maximum.\footnote{More concretely, we assume the track $\gamma$ intersects its evolute (the locus of centers of curvature) transversely, away from its cusps.}  Otherwise, we would need higher derivatives to disambiguate the puppy's behavior.  In the extreme case where $\gamma$ contains both an open circular arc $\alpha$ and its center $c$, all configurations where $h=c$ and $p\in \alpha$ are stable.

\subsection{Attraction diagrams}
\label{sec:puppydiagram}

The \EMPH{attraction diagram} of the track $\gamma$ is a decomposition of the configuration space $S^1\times S^1$ by critical configurations.  Our genericity assumptions imply that the set of critical points---the common boundary of the forward and backward configurations---is the union of a finite number of disjoint simple closed curves, which we call \emph{critical cycles}.  At least one of these critical cycles, the main diagonal $x=y$, consists entirely of stable configurations; critical cycles can also consist entirely of unstable configurations.  If a critical cycle is neither entirely stable nor entirely unstable, then its points of vertical tangency are pivot configurations, and these points subdivide the curve into $x$-monotone paths, which alternately consist of stable and unstable configurations.

\cref{F:diagram} shows a sketch of the attraction diagram of a simple closed curve.  We visualize the configuration torus $S^1\times S^1$ as a square with opposite sides identified.  Green and red paths indicate stable and unstable configurations, respectively; blue dots indicate pivot configurations; and backward configurations are shaded light gray.  \cref{F:spiral-diagram} shows the attraction diagram for a more complex polygonal track, with slightly different coloring conventions.  (Again, we will discuss polygonal tracks in more detail in \cref{S:polygons}.)

\begin{figure}[ht]
\centering
\raisebox{-0.5\height}{\includegraphics[scale=0.6,page=1]{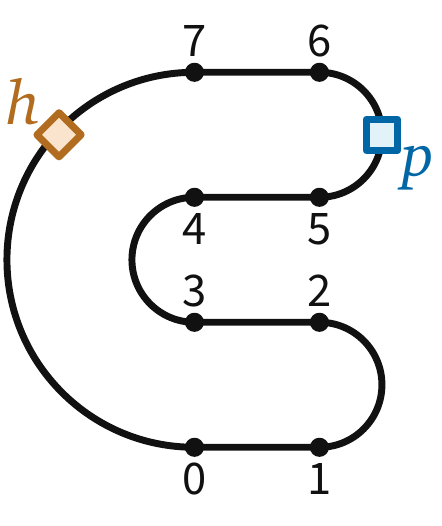}}
\qquad
\raisebox{-0.5\height}{\includegraphics[scale=0.5,page=3]{Fig/square-puppy-diagram}}
\caption{The attraction diagram of a simple closed curve, with one unstable critical configuration emphasized.}
\label{F:diagram}
\end{figure}

\begin{figure*} 
\centering
\includegraphics[scale=0.5]{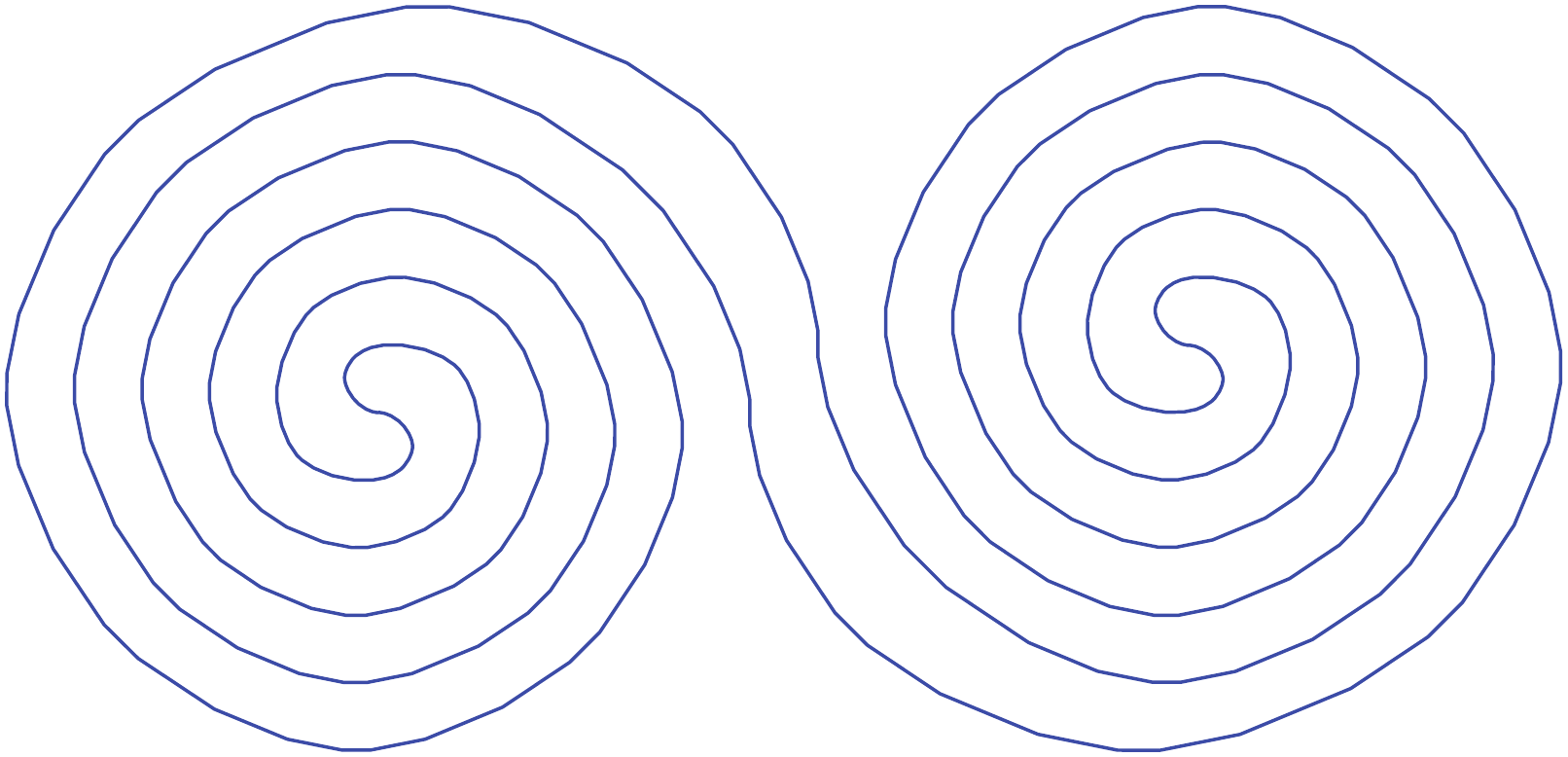}\\[3ex]
\includegraphics[width=\linewidth]{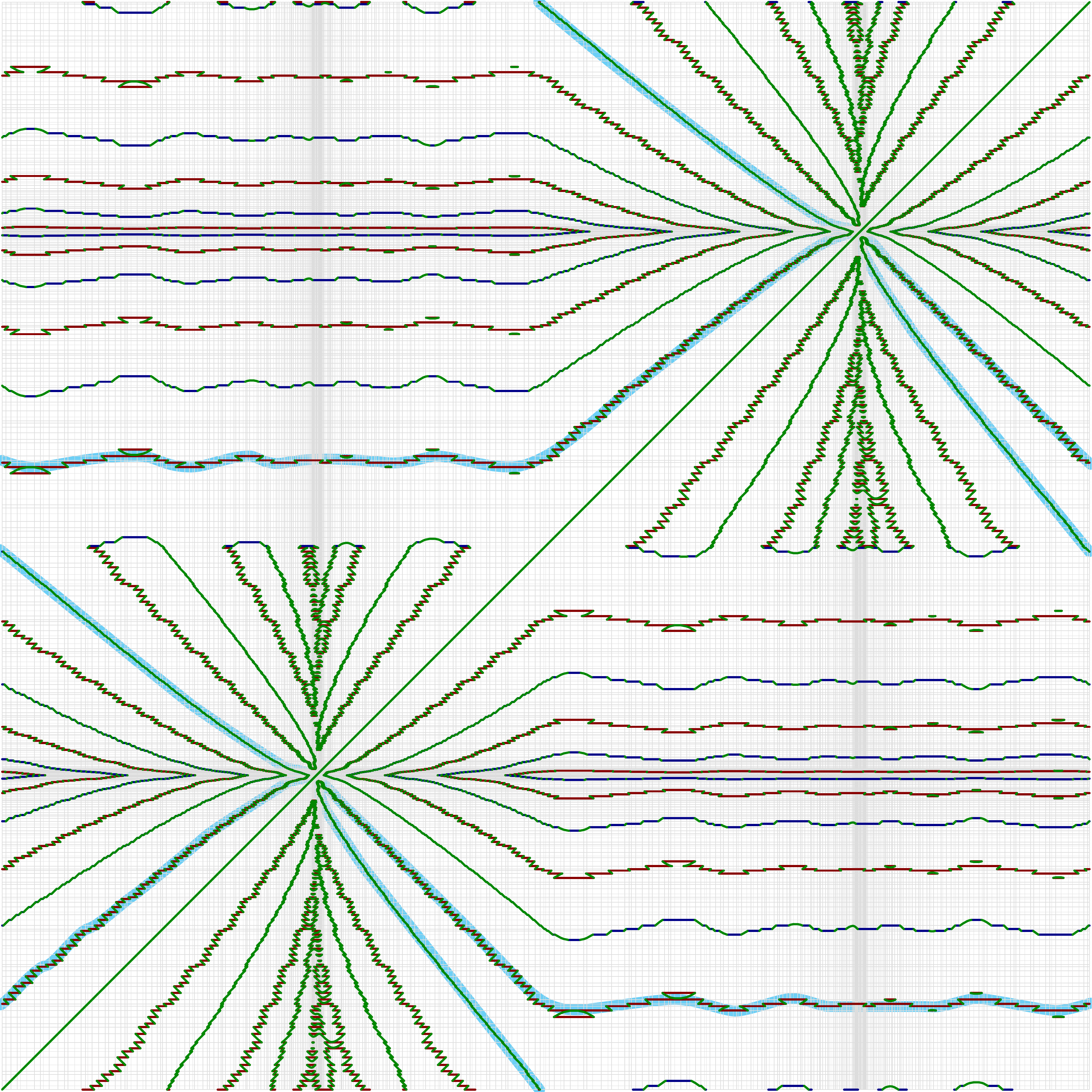}
\caption{The attraction diagram of a complex simple polygon.  Serrations in the diagram are artifacts of the curve being polygonal instead of smooth.  The river is highlighted in blue.}
\label{F:spiral-diagram}
\end{figure*}

The cycles in any attraction diagram have a simple but important topological structure.  
A critical cycle in the attraction diagram is \emph{contractible} if it is the boundary of a simply connected subset of the torus $S^1\times S^1$ and \emph{essential} otherwise.  For example, the main diagonal is essential, and the attraction diagram in \cref{F:diagram} contains two contractible critical cycles and two essential critical cycles.

\begin{lemma}
\label{L:even}
The attraction diagram of any generic closed curve contains an even number of essential critical cycles.
\end{lemma}

\begin{proof}
This lemma follows immediately from standard homological arguments, but for the sake of completeness we sketch a self-contained proof.

Fix a generic closed curve $\gamma$.  Let $\alpha$ be the horizontal cycle $\{(0, y) \mid y\in S^1\}$, and let $\beta$ be the vertical cycle $\{(x, 0) \mid x\in S^1\}$ in the torus $S^1\times S^1$. 
 Without loss of generality, assume $\alpha$ and $\beta$ intersect every critical cycle in the attraction diagram of $\gamma$ transversely.

A critical cycle $C$ in the attraction diagram is contractible if and only if $\alpha$ and $\beta$ each  cross $C$ an even number of times.  (Indeed, this parity condition characterizes all simple contractible closed curves in the torus.)  
On the other hand, $\alpha$ and $\beta$ each cross the main diagonal once.  
It follows that $\alpha$ and $\beta$ each cross \emph{every} essential critical cycle an odd number of times;
otherwise, some pair of essential critical cycles would intersect.

Because the critical cycles are the boundary between the forward and backward configurations, $\alpha$ and $\beta$ each contain an even number of critical points.  The lemma now follows immediately.
\end{proof}

We emphasize that this lemma does \emph{not} actually require the track $\gamma$ to be simple; the argument relies only on properties of generic functions over the torus that are minimized along the main diagonal.

\subsection{Dual attraction diagrams}

Our analysis also relies on a second diagram, which we call the \EMPH{dual attraction diagram} of the track.  
We hope the following  intuition is helpful. 
While the attraction diagram tells us the possible positions of the puppy depending on the position
of the human, the dual attraction diagram gives us the possible positions of the human depending on the position of the puppy.
For each puppy configuration~$y\in S^1$, we consider the 
normal line $\normal(y)$. 
We are interested in the intersection points of $\gamma$ with $\normal(y)$,
as those are the possible positions of the human.
The idea of the dual attraction diagram is to trace the positions of the human as a function of the position of the puppy, see \cref{F:human-diagram}.

Let $\tangent(y)$ denote the directed line tangent to $\gamma$ at the point $\gamma(y)$.  For any configuration $(x,y)$, let \EMPH{$\ell(x,y)$} denote the the distance from $\gamma(x)$ to the tangent line $\tangent(y)$, signed so that $\ell(x,y)>0$ if the human point $\gamma(x)$ lies to the left of $\tangent(y)$ and $\ell(x,y)<0$ if $\gamma(y)$ lies to the right of $\tangent(y)$.   More concisely, assuming without loss of generality that the track $\gamma$ is parameterized by arc length, $\ell(x,y)$ is twice the signed area of the triangle with vertices $\gamma(x)$, $\gamma(y)$, and $\gamma(y) + \gamma'(y)$.

Let $L\colon S^1\times S^1\to S^1\times \Real$ denote the function $L(x,y) = (y, \ell(x,y))$.
The dual attraction diagram is the decomposition of the infinite cylinder $S^1\times \Real$ by the points $\set{ L(x,y) \mid \text{$\strut(x, y)$ is critical}}$.  At the risk of confusing the reader, we refer to the image $L(x,y)\in S^1\times \Real$ of any critical configuration $(x,y)$ as a critical point of the dual attraction diagram.

The dual attraction diagram can also be described as follows.  For any $y\in S^1$ and $d\in \Real$, let $\Gamma(y, d)$ denote the point on the normal line $\normal(y)$ at distance $d$ to the left of the tangent vector $\gamma'(y)$.  More formally, assuming without loss of generality that $\gamma$ is parametrized by arc length, we have $\Gamma(y, d) = \gamma(y) + d \left[\begin{smallmatrix} 0 & -1 \\ 1 & 0 \end{smallmatrix}\right] \gamma'(y)$.
We emphasize that $\Gamma(y,d)$ does not necessarily lie on the curve $\gamma$.   The dual attraction diagram is the decomposition of the cylinder $S^1\times\Real$ by the preimage $\Gamma^{-1}(\gamma)$ of $\gamma$. 

\begin{figure}[ht]
\centering
\includegraphics[scale=0.5]{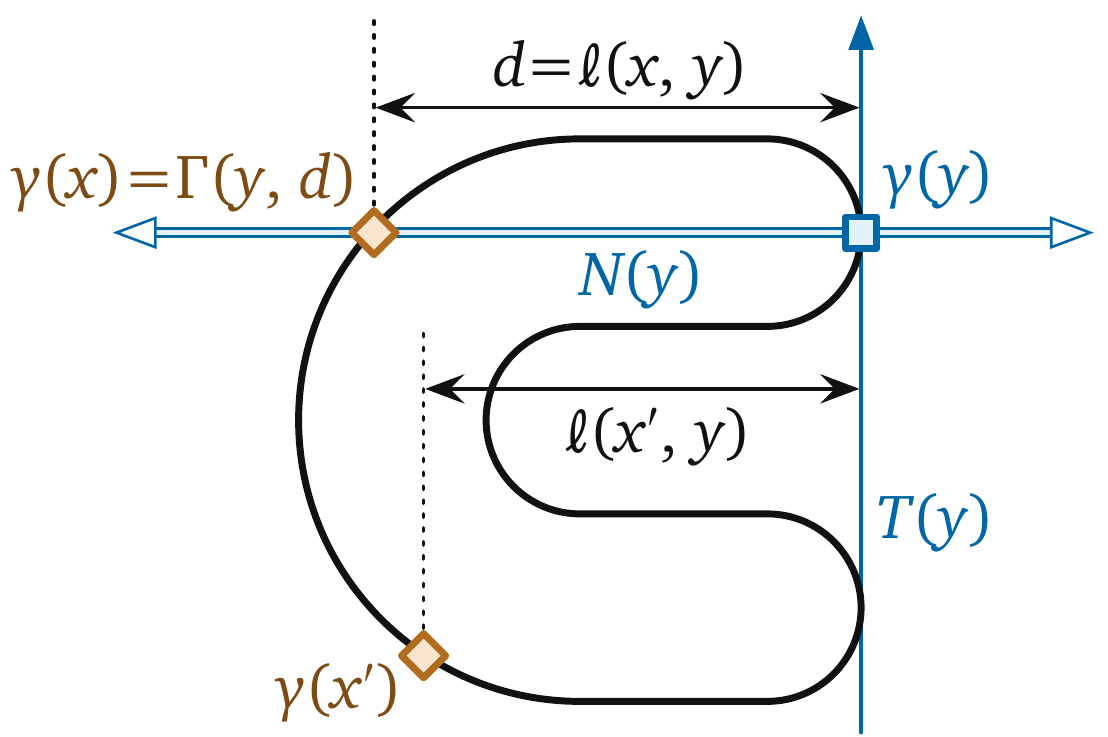}
\caption{Examples of the functions $\ell$ and $\Gamma$ used to define the dual attraction diagram.}
\label{F:human-parameters}
\end{figure}

Because $\gamma$ is simple and regular, the dual attraction diagram is the union of simple disjoint closed curves.  
The function $L$ continuously maps each critical cycle in the attraction diagram to a closed curve in the cylinder $S^1\times \Real$.
Thus, the restriction of $L$ to the set of critical configuration is a homeomorphism onto its image in the dual attraction diagram.  
In particular, $L$ maps  the main diagonal $x=y$ to the horizontal axis $\ell(x,y) = 0$ of the dual attraction diagram.  
We emphasize, however, that the two diagrams are not topologically  equivalent.  
\cref{F:human-diagram} shows the dual attraction diagram of the same track whose attraction diagram is shown in \cref{F:diagram}; here preimages of points inside the track are shaded.

\begin{figure}[ht]
\centering
\raisebox{-0.5\height}{\includegraphics[scale=0.5,page=1]{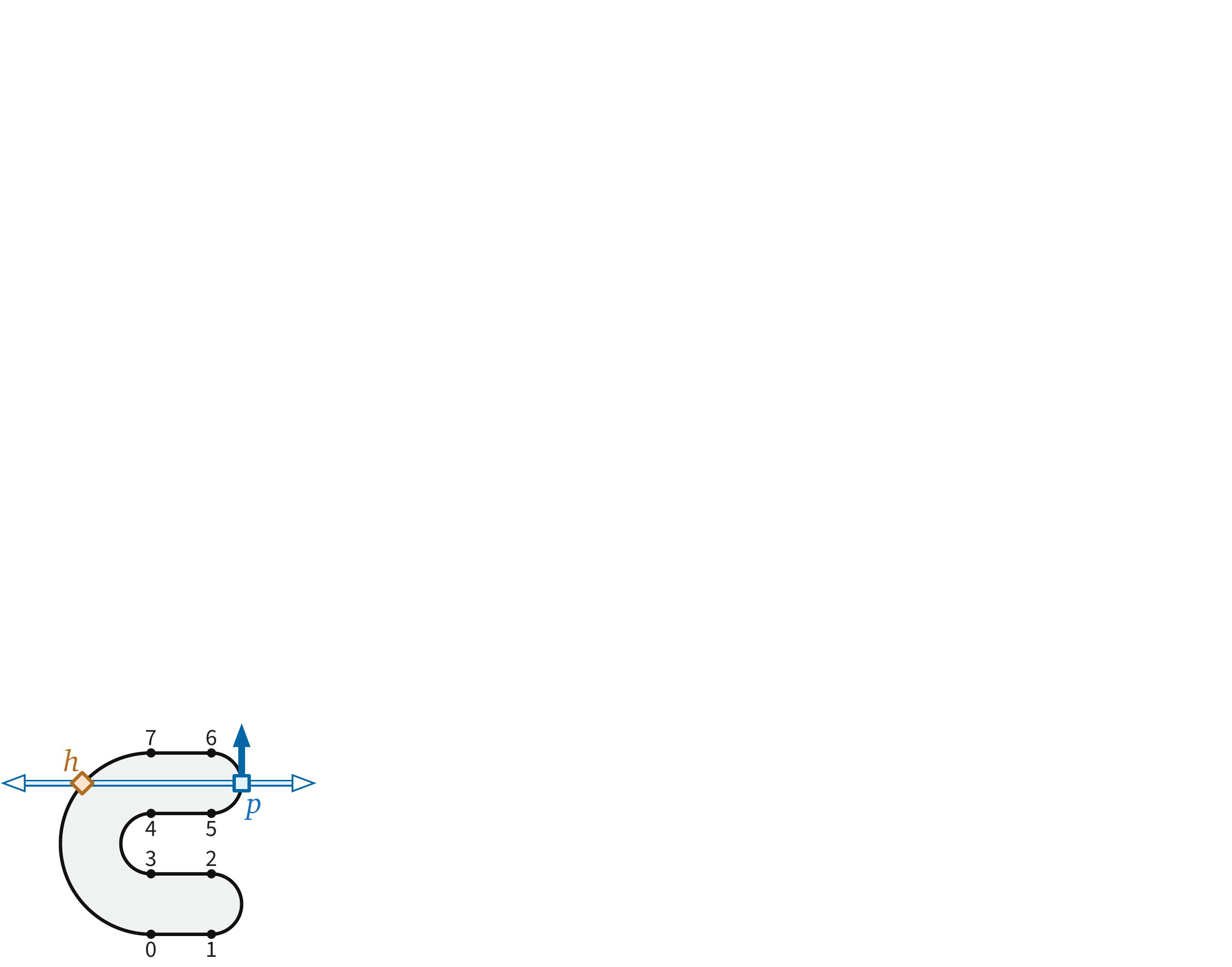}}
\raisebox{-0.5\height}{\includegraphics[scale=0.45,page=2]{Fig/normal-puppy}}
\caption{The dual attraction diagram of a simple closed curve, with one critical configuration emphasized.  Compare with \cref{F:diagram}.}
\label{F:human-diagram}
\end{figure}

\begin{lemma}
For any generic simple closed curve $\gamma$, the attraction diagram of $\gamma$ and the dual attraction diagram of $\gamma$ contain the same number of essential critical cycles.
\end{lemma}

\begin{proof}
Let $\alpha$ denote the horizontal cycle $y=0$ in the torus $S^1\times S^1$, and let $\alpha'$ be the vertical line $y=0$ in the infinite cylinder $S^1\times \Real$.  Let $C$ be any critical cycle on the attraction diagram, and let $C' = L(C)$ be the corresponding critical cycle in the dual attraction diagram. 

Recall from the proof of \cref{L:even} that $C$ is contractible on the torus if and only if $\abs{C\cap \alpha}$ is even.  Similarly, $C'$ is contractible in the cylinder if and only if $\abs{C'\cap \alpha'}$ is even.  The map $L\colon S^1\times S^1\to S^1\times \Real$ maps $C\cap \alpha$ bijectively to $C'\cap \alpha'$.  We conclude that $C$ is essential if and only if $C'$ is essential.
\end{proof}

With this correspondence in hand, we can now more carefully describe the topological structure of the \emph{attraction} diagram when the track is simple.

\begin{lemma}
\label{L:simple-good}
The attraction diagram of a \textbf{simple} generic closed curve contains \textbf{two} essential critical cycles.
\end{lemma}

\begin{proof}
Fix a generic closed curve $\gamma$.  \cref{L:even} implies that the attraction diagram of $\gamma$ contains at least two essential critical cycles, one of which is the main diagonal.  Thus, to prove the lemma, it remains to show that there are \emph{at most} two essential critical cycles, in either the attraction diagram or the dual attraction diagram.

Let $\Sigma \subset S^1\times \Real$ denote the set of essential critical cycles in the \emph{dual attraction} diagram.  Any two cycles in $\Sigma$ are homotopic---meaning one can be continuously deformed into the other---because there is only one nontrivial homotopy class of simple cycles on the infinite cylinder $S^1\times \Real$.  It follows that the cycles in $\Sigma$ have a well-defined vertical total order.  In particular, the highest and lowest intersection points between any vertical line and $\Sigma$ always lie on the \emph{same} two essential cycles in $\Sigma$.

Without loss of generality, suppose $\gamma(0)$ is a point on the convex hull of $\gamma$ with a unique tangent line.  Let $C$ be any essential critical cycle in the attraction diagram of $\gamma$, and let $C' = L(C)$ denote the corresponding essential cycle in the dual attraction diagram.  $C$ must pass through all possible puppy positions \emph{and} all possible human positions; thus, $C$  contains a configuration $(0,y)$ for some parameter $y\in S^1$.  Recall that $\normal(y)$ denotes the line normal to $\gamma$ at $\gamma(y)$.  Then $\gamma(0)$ must also lie on the convex hull of $\gamma\cap \normal(y)$.  We conclude that $C'$ must be either the highest or lowest essential critical cycle in the dual attraction diagram.  We conclude that there are at most two critical cycles, completing the proof.
\end{proof}

In the rest of the paper, we mnemonically refer to the two essential critical cycles in the attraction diagram of a simple track as the \EMPH{main diagonal} and the \EMPH{river}.

We emphasize that the converse of \cref{L:simple-good} is false; there are non-simple tracks whose attraction diagrams have exactly two essential critical cycles.  (Consider the figure-eight curve $\infty$.)  Moreover, we conjecture that \cref{L:simple-good} can be generalized to all (smooth) tracks with turning number $\pm 1$.

\section{Dexter and sinister strategies}
\label{S:dexterity}

We can visualize any strategy for the human to catch the puppy as a path through the attraction diagram that consists entirely of segments of stable critical paths and vertical segments, as shown in \cref{F:catch}.  We refer to the vertical segments as \emph{pivots}.  Every pivot (except possibly the first) starts at a pivot configuration, and every pivot ends at a stable configuration.

\begin{figure}[ht]
\centering
\raisebox{-0.5\height}{\includegraphics[scale=0.4]{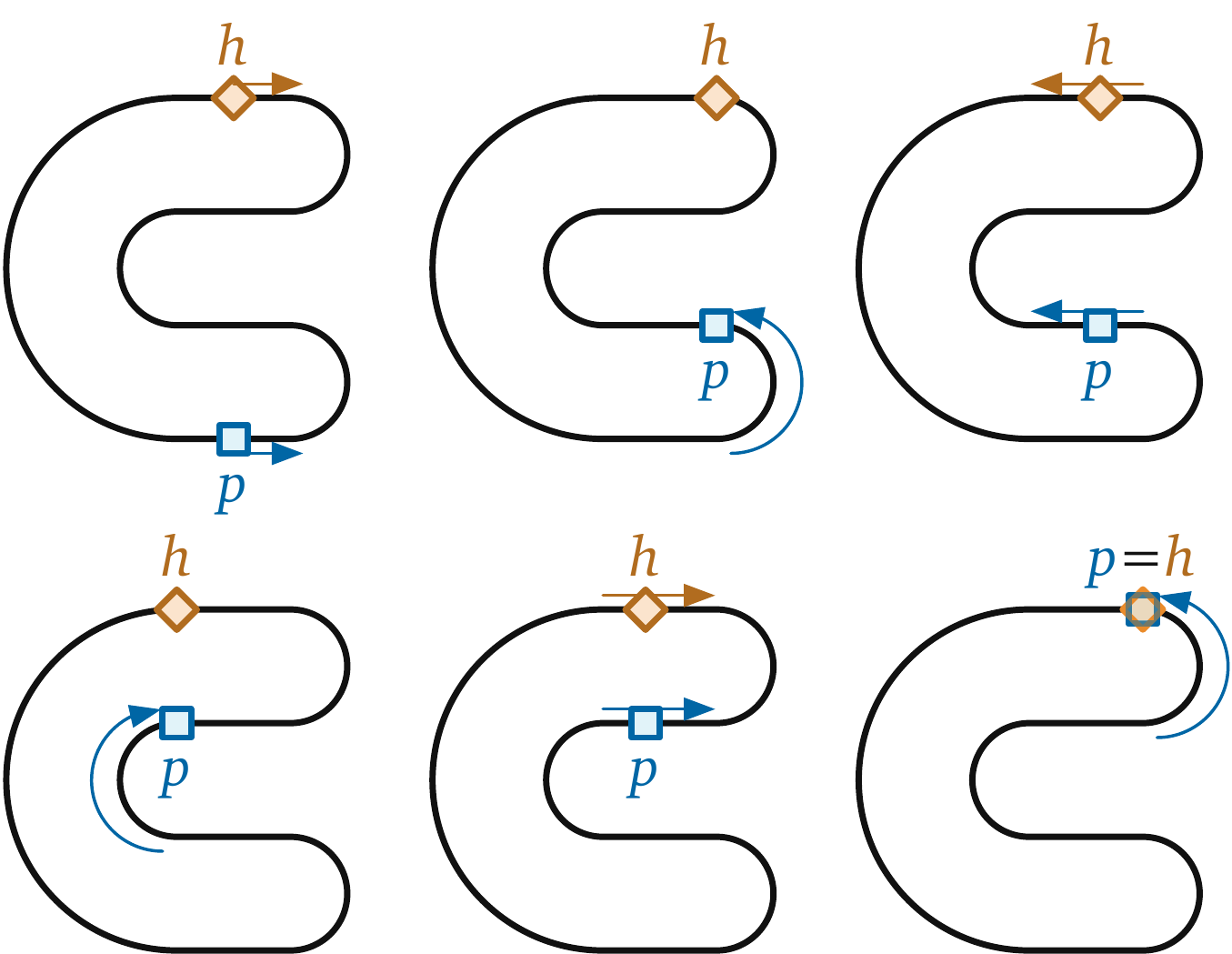}}\hfil
\raisebox{-0.5\height}{\includegraphics[scale=0.5]{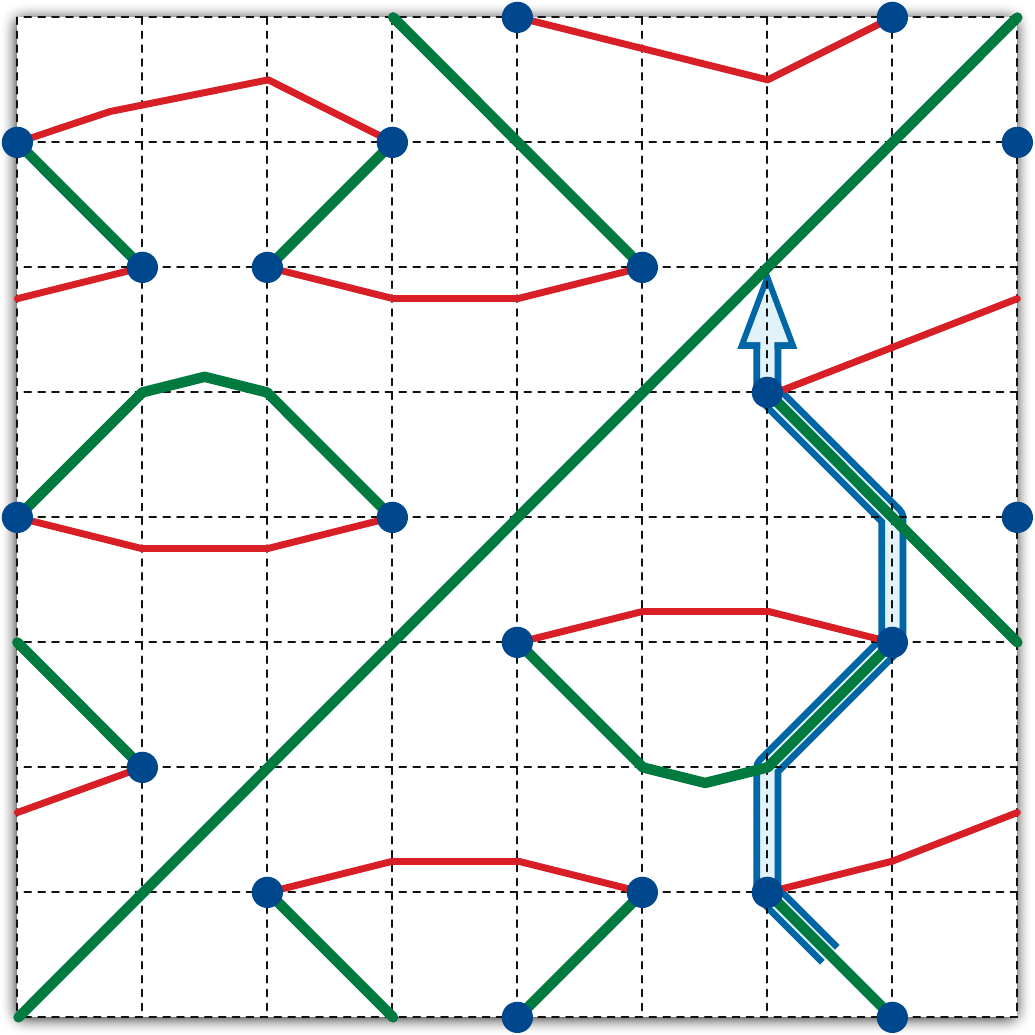}}
\caption{A sinister strategy for catching the puppy; compare with \cref{F:init-example,F:diagram}.}
\label{F:catch}
\end{figure}

We call a strategy \EMPH{dexter} if it ends with a backward pivot---a \emph{downward} segment, approaching the main diagonal to the \emph{right}---and we call a configuration $(x,y)$ \emph{dexter} if there is a dexter strategy for catching the puppy starting at $(x,y)$.  Similarly, a strategy is \EMPH{sinister} if it ends with a forward pivot---a \emph{skyward} segment, approaching the main diagonal to the \emph{left}---and a configuration is sinister if it is the start of a sinister strategy.\footnote{\emph{Dexter} and \emph{sinister} are Latin for right (or skillful, or fortunate, or proper, from a Proto-Indo-European root meaning “south”) and left (or unlucky, or unfavorable, or malicious), respectively.}  A single configuration can be both dexter and sinister; see \cref{F:dexter}.

\begin{figure}[ht]
\centering
\includegraphics[scale=0.5,page=3]{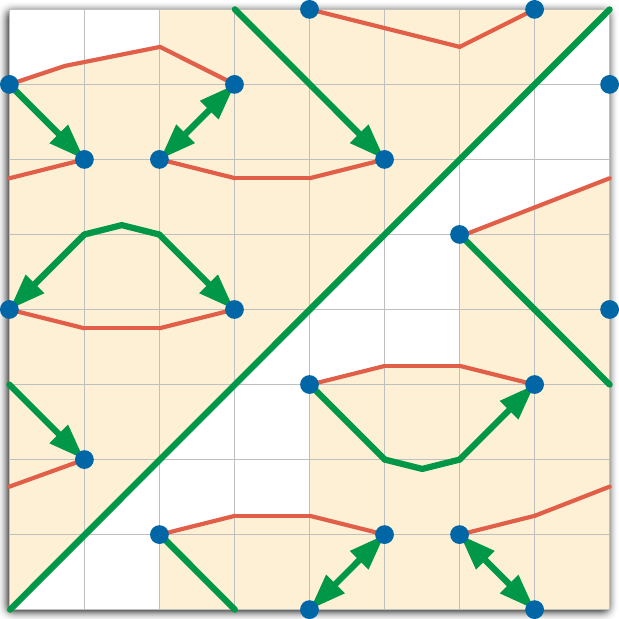}\hfil
\includegraphics[scale=0.5,page=2]{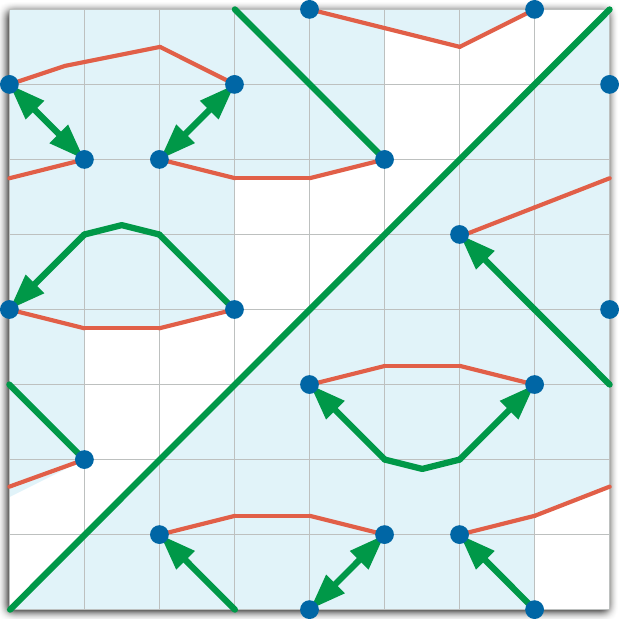}
\caption{Dexter (orange) and sinister (cyan) configurations in the example attraction diagram.  Arrows on the stable critical paths describe dexter and sinister strategies for catching the puppy.}
\label{F:dexter}
\end{figure}

\begin{theorem}
\label{Th:good-catch}
Let $\gamma$ be a generic track whose attraction diagram has exactly two essential critical cycles.  Every configuration on $\gamma$ is either dexter or sinister; thus, the human can catch the puppy on $\gamma$ from any starting configuration.
\end{theorem}

Before giving the proof, we emphasize that \cref{Th:good-catch} does not require  the track $\gamma$ to be simple.  Also, it is an open question whether having exactly two essential critical cycle curves is a \emph{necessary} condition for the human to always be able to catch the puppy.  (We conjecture that it is not.)

\begin{proof}
Fix a generic track $\gamma$ whose attraction diagram has exactly two essential critical cycles, which we call the \emph{main diagonal} and the \emph{river}.  Assume $\gamma$ has at least one pivot configuration, since otherwise, from any starting configuration, the puppy runs directly to the human.

Let $D$ be the set of all dexter configurations, and let $S$ be the set of all sinister configurations.  We claim that $D$ and $S$ are both annuli that contain both the main diagonal and the river.  Because $S$ and $D$ meet on opposite sides of the main diagonal, this claim implies that $D\cup S$ is the entire torus, completing the proof of the lemma.  We prove our claim explicitly for $D$; a symmetric argument establishes the claim for $S$.

\medskip
For purposes of argument, we partition the attraction diagram of $\gamma$ by extending vertical segments from each pivot configuration to the next critical cycles directly above and below.  We call the cells in this decomposition \emph{trapezoids}, even though their top and bottom boundaries may not be straight line segments.  At each forward pivot configuration $p$, we color the vertical segment above $(x,y)$ \emph{green} and the vertical segment below $p$ \emph{red}; the colors are reversed for backward vertical segments, see \cref{F:trapezoids}.

The first step of any strategy is a (possibly trivial) pivot onto a stable critical path.  Because the human and puppy can move freely within any stable critical path $\sigma$, either every point in $\sigma$ is dexter, or no point in $\sigma$ is dexter.  Similarly, for any green pivot segment $\pi$, either every point in $\pi$ is dexter or no point in $\pi$ is dexter.

Consider any trapezoid $\tau$, and let $\sigma$ be the stable critical path on its boundary.  Starting in any configuration in $\tau$, the puppy immediately moves to a configuration on $\sigma$.  Thus, if any point in $\tau$ is dexter, then $\sigma$ is dexter, which implies that \emph{every} point in $\tau$ is dexter.  Thus, we can describe entire trapezoids as dexter or not dexter.  It follows that $D$ is the union of trapezoids.

If two trapezoids share a stable critical path \emph{other than the main diagonal}, then either both trapezoids are dexter or neither is dexter.  Similarly, if the green pivot segment leaving a pivot configuration $p$ is dexter, then all four trapezoids incident to $p$ are dexter; otherwise, either two or none of these four trapezoids are dexter.

We conclude that aside from the main diagonal, the boundary of $D$ consists entirely of unstable critical paths, pivot configurations, and red vertical segments.  Moreover, for every pivot configuration $p$ on the boundary of~$D$, the green pivot segment leaving $p$ is \emph{not} dexter.

\begin{figure}[ht]
\centering
\includegraphics[scale=0.5]{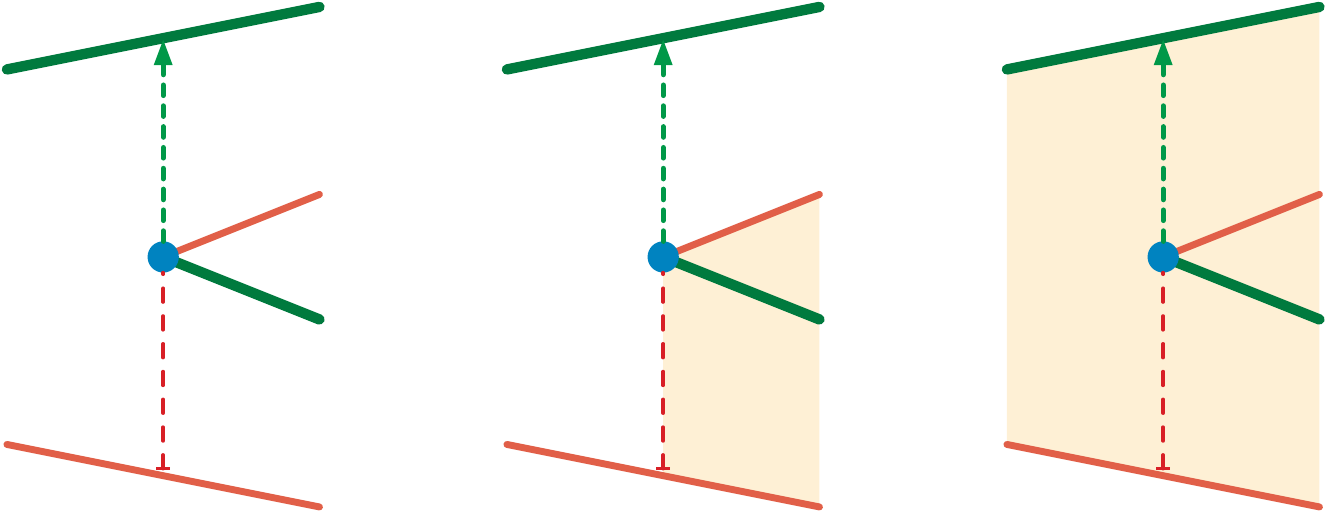}
\caption{Possible arrangements of dexter trapezoids near a forward pivot configuration.}
\label{F:trapezoids}
\end{figure}

By definition, every point in $D$ is connected by a (dexter) path to the main diagonal, so $D$ is non-empty and connected.  On the other hand, $D$ excludes a complete cycle of forward configurations just below the main diagonal.  For any $x \in S^1$, let $D(x)$ denote the set of dexter configurations $(x, y)$; this set consists of one or more vertical line segments in the attraction diagram.

Suppose for the sake of argument that some set $D(x)$ is disconnected.  Because~$D$ is connected, the boundary of $D$ must contain a \emph{concave vertical bracket}: A vertical boundary segment $\pi$ whose adjacent critical boundary segments both lie (without loss of generality) to the right of $\pi$, but $D$ lies locally to the left of $\pi$.  See \cref{F:bracket}.  Let $p$ be the pivot configuration at one end of $\pi$.  The green vertical segment on the other side of $p$ is dexter, which implies that \emph{all} trapezoids incident to $p$ are dexter, contradicting the assumption that $\pi$ lies  on the boundary of $D$.  We conclude that for all $x$, the set $D(x)$ is a single vertical line segment; in other words, $D$ is a \emph{monotone} annulus.

\begin{figure}[ht]
\centering
\includegraphics[scale=0.5]{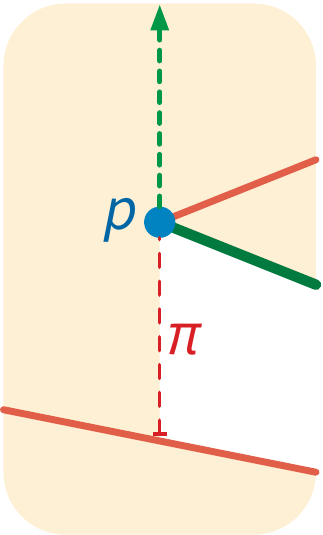}
\caption{A hypothetical concave vertical bracket on the boundary of $D$.}
\label{F:bracket}
\end{figure}

The bottom boundary of $D$ is the main diagonal.  The monotonicity of $D$ implies that the top boundary of $D$ is a monotone “staircase” alternating between upward red vertical segments and rightward unstable critical paths.  Every trapezoid immediately above the top boundary of $D$ contains only forward configurations.  Thus, there is a complete essential cycle $\phi$ of forward configurations just above the upper boundary of $D$.  Because $\phi$ contains only forward configurations, $\phi$ must lie entirely above the river.  It follows that $D$ contains the entire river. 

Symmetrically, $S$ is an annulus bounded above by the main diagonal and bounded below by a non-contractible cycle of backward configurations; in particular, the entire river lies inside $S$.  We conclude that $D\cup S$ is the entire configuration torus.
\end{proof}

If the attraction diagram of $\gamma$ has more than two essential critical cycles curves, then $D$ and $S$ are still monotone annuli, each bounded by the main diagonal and an essential cycle of red vertical segments and unstable paths, and thus $S$ and $D$ each contain at least one essential critical cycle other than the main diagonal.  However, $D\cup S$ need not cover the entire torus.

\begin{corollary}
The human can catch the puppy on any generic simple closed track, from any starting configuration.
\end{corollary}

\section{Polygonal tracks}
\label{S:polygons}

Our previous arguments require, at a minimum, that the track has a continuous derivative that is never equal to zero. 
We now extend our results to polygonal tracks, which do not have well-defined tangent directions at their vertices.

\subsection{Polygonal attraction diagrams}

Throughout this section, we fix a simple polygonal track $P$ with $n$ vertices.
We regard $P$ as a continuous piecewise-linear function $P\colon S^1\to \Real^2$, parametrized by arc length.  Without loss of generality $P(0)$ is a vertex of the track.  
We index the vertices and edges of $P$ in order, starting with $v_0 = P(0)$, where edge $e_i$ connects $v_i$ to $v_{i+1}$; all index arithmetic is implicitly performed modulo $n$.

To properly describe the puppy's behavior, we must also account for the direction that the puppy is facing, even when the puppy lies at a vertex.  To that end, we represent the track using both a continuous \emph{position} function $\pi\colon S^1\to \Real$ and a continuous \emph{direction} function $\theta\colon S^1\to S^1$, such that for all $y\in S^1$, the derivative vector $\pi'(y)$ is a non-negative scalar multiple of the unit vector $\theta(y)$.  
Intuitively, as we increase~$y$, the puppy alternately moves at constant speed along edges (when $\pi'(y)$ is a positive multiple of $\theta(y)$) and continuously turns at constant speed at vertices (when $\pi'(y) = 0$).

We classify any human-puppy configuration $(x,y) \in S^1\times S^1$ as \emph{forward}, \emph{backward}, or \emph{critical}, if the dot product $(P(x) - \pi(y))\cdot \theta(y)$ is negative, positive, or zero, respectively.  In any forward configuration $(x,y)$, the puppy moves to increase the parameter $y$; in any backward configuration, the puppy moves to decrease the parameter $y$.  (The human's direction is irrelevant.)  The \emph{attraction diagram} is the set of all critical configurations $(x,y)\in S^1\times S^1$.  We further classify critical configurations $(x,y)$ as follows:
\begin{itemize}\itemsep0pt
\item \emph{final} if $P(x) = \pi(y)$,
\item \emph{stable} if $(x, y-\e)$ is forward and $(x, y+\e)$ is backward for all suffic. small $\e>0$,
\item \emph{unstable} if $(x, y-\e)$ is backward and $(x, y+\e)$ is forward for all suffic. small $\e>0$,
\item \emph{forward pivot} if $(x, y-\e)$ and $(x, y+\e)$ are both forward for all suffic. small $\e>0$, or 
\item \emph{backward pivot} if $(x, y-\e)$ and $(x, y+\e)$ are both backward for all suffic. small $\e>0$.
\end{itemize}
Straightforward case analysis implies that this classification is exhaustive.  

To define the attraction diagram of $P$, we decompose the torus $S^1\times S^1$ into a $2n\times n$ grid of rectangular cells, where each column corresponds to an edge $e_j$ containing the human, and each row corresponds to either a vertex $v_i$ or an edge $e_i$ containing the puppy.  
The \emph{main diagonal} of the attraction diagram is the set of all final configurations. Strictly speaking, in this case the ``main diagonal'' is not just a straight line, but consists of alternating diagonal and vertical segments.
We can characterize the critical points inside each cell as follows:

Each edge-edge cell $e_i\times e_j$ contains at most one boundary-to-boundary path of stable critical configurations $(x,y)$. 
Refer to \cref{F:edge-edge}.

\begin{figure}[ht]
\centering
\includegraphics[scale=0.4]{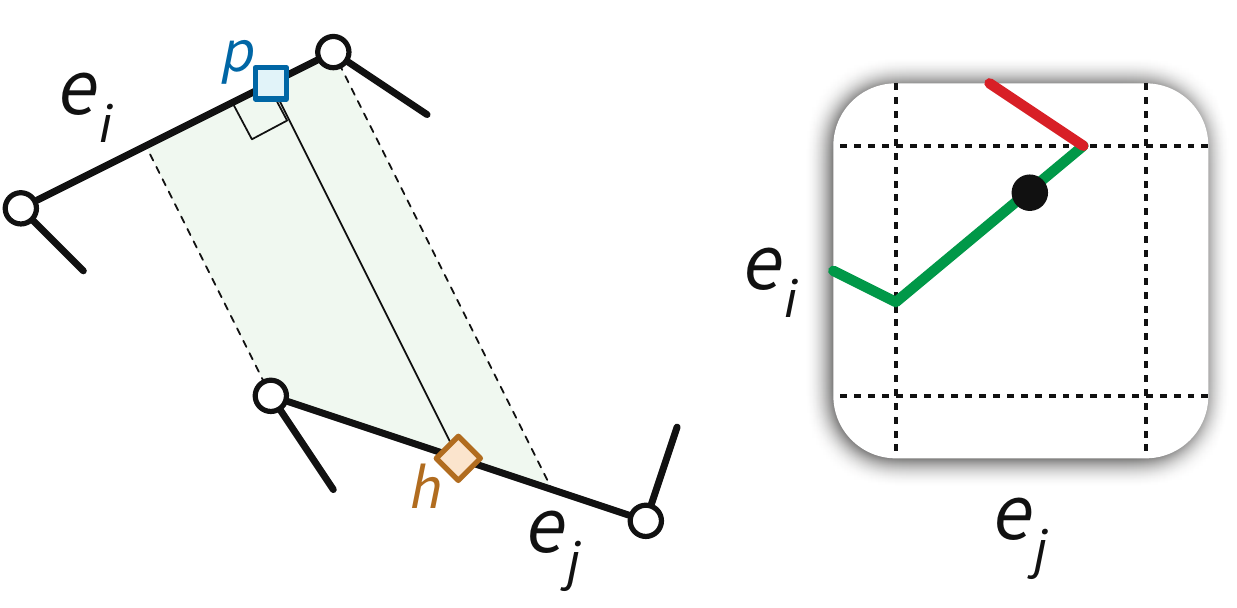}
\caption{All edge-edge critical configurations are stable.}
\label{F:edge-edge}
\end{figure}

Each vertex-edge cell $v_i\times e_j$ contains at most one boundary-to-boundary path of stable critical configurations and at most one boundary-to-boundary path of unstable critical configurations.  If the cell contains both paths, they are disjoint.  A configuration $(x,y)$ with $\pi(y) = v_i$ is stable if and only if $P(x)$ lies in the outer normal cone at $v_i$, and unstable if and only if $P(x)$ lies in the inner normal cone at $v_i$; see \cref{F:vertex-edge}.  

\begin{figure}[ht]
\centering
\includegraphics[scale=0.4]{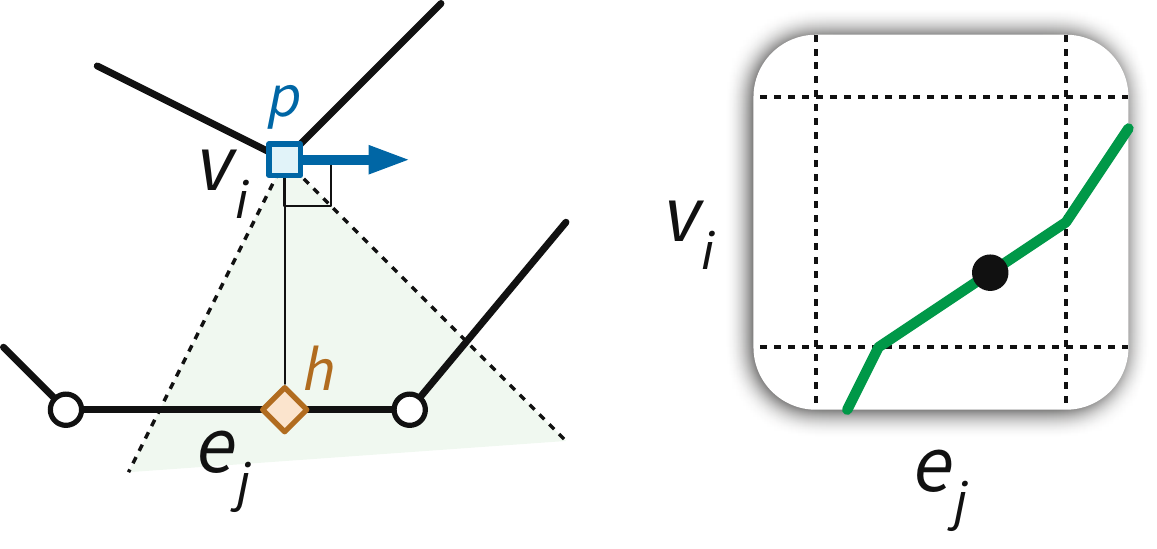}\qquad\qquad
\includegraphics[scale=0.4]{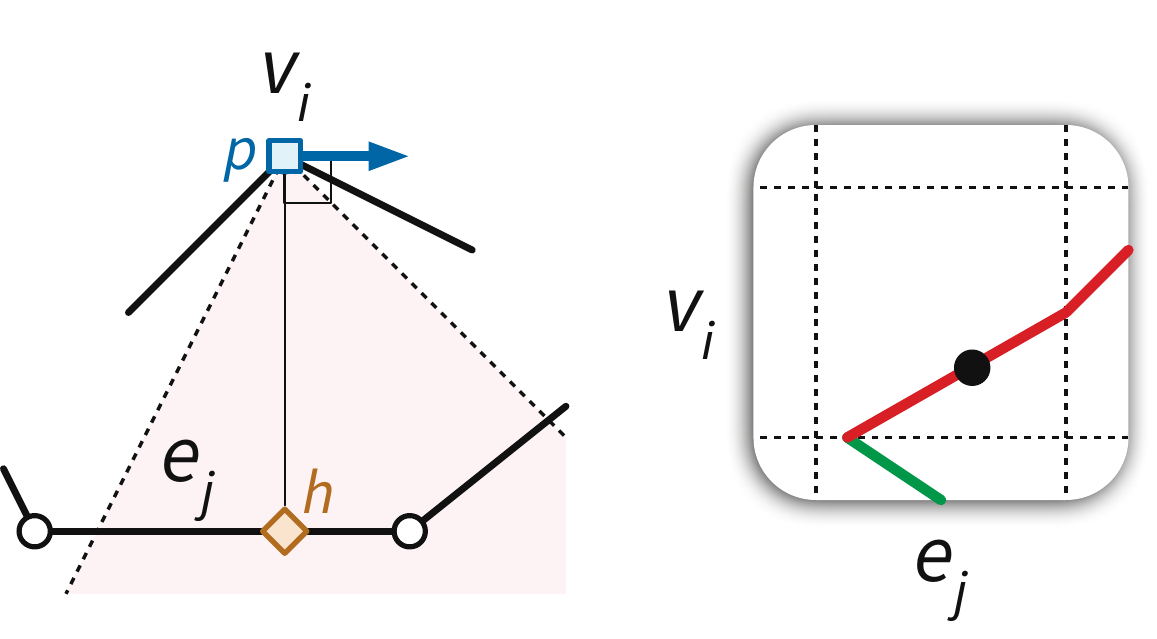}
\caption{Stable and unstable vertex-edge critical configurations.}
\label{F:vertex-edge}
\end{figure}

\subsection{Polygonal pivot configurations}

Unlike the attraction diagrams of generic smooth curves defined in \cref{sec:puppydiagram}, the attraction diagrams of polygons are not always well-behaved.  In particular, a pivot configuration may be incident to more (or fewer) than two critical curves, and in extreme cases, pivot configurations need not even be discrete.  We call such a configuration a \emph{degenerate} pivot configuration.

In any pivot configuration $(x,y)$, the puppy $\pi(y)$ lies at some vertex $v_i$, the  puppy's direction $\theta(y)$ is parallel to either $e_i$ (or $e_{i+1}$).
Generically, each pivot configuration is a shared endpoint of an unstable critical path in cell $v_i\times e_j$ and a stable critical path in cell $e_i\times e_j$ (or $e_{i-1}\times e_j$); see \cref {F:vertex-edge-pivot}.

\begin{figure}[ht]
\centering
\includegraphics[scale=0.4]{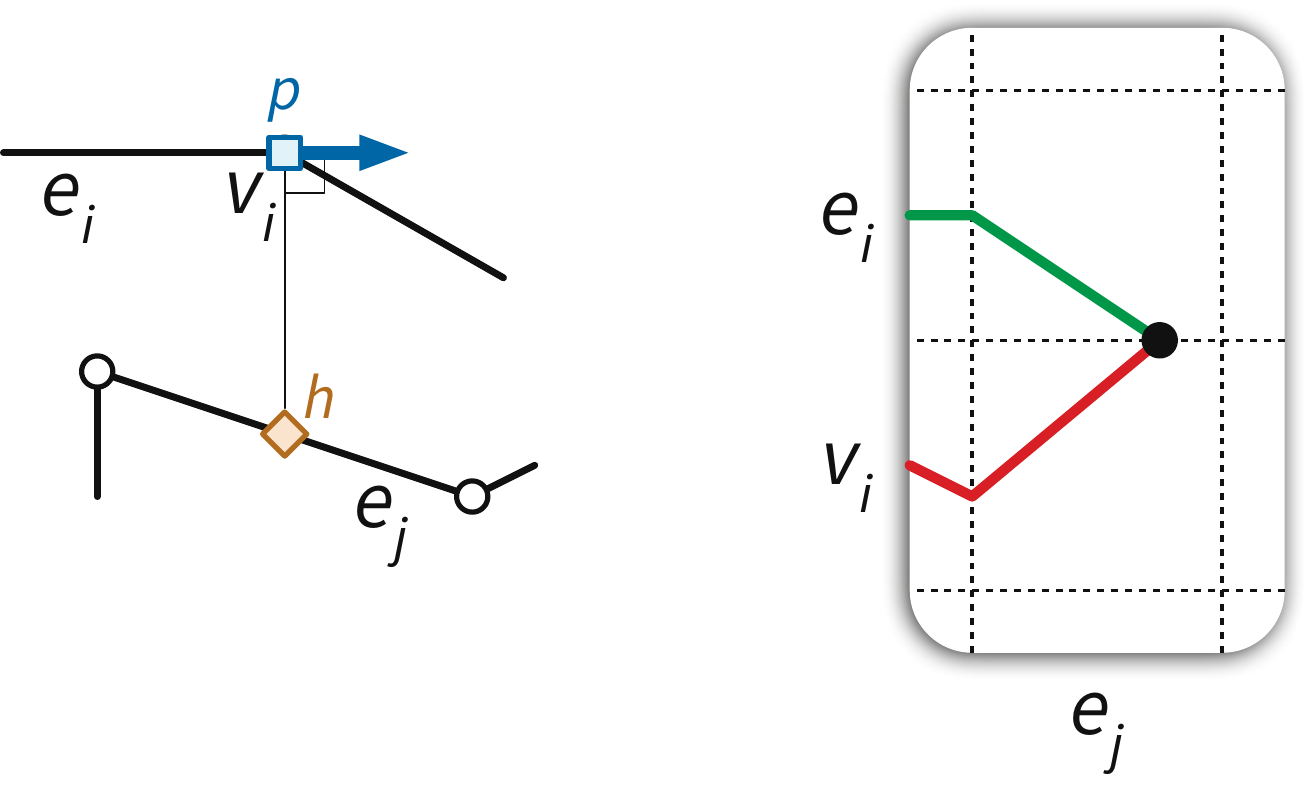}
\caption{Near a non-degenerate pivot configuration.}
\label{F:vertex-edge-pivot}
\end{figure}

There are three distinct ways in which degenerate pivot configurations can appear.

A \EMPH{type-1 degeneracy} is caused by an acute angle on $P$.
Specifically,
let $v_i$ be a vertex of $P$.
The configuration $(x, y)$ with $P(x) = \pi(y) = v_i$ is degenerate if the angle between $e_{i-1}$ and $e_i$ is strictly acute.
In the attraction diagram of a type-1 degeneracy, two stable critical curves and two unstable critical curves end on a single vertical section of the main diagonal (corresponding to the human and the puppy being both at $v_i$, but the puppy facing in different directions). 
Refer to \cref{F:degen1}.

\begin{figure}[ht]
\centering
\includegraphics[scale=0.4]{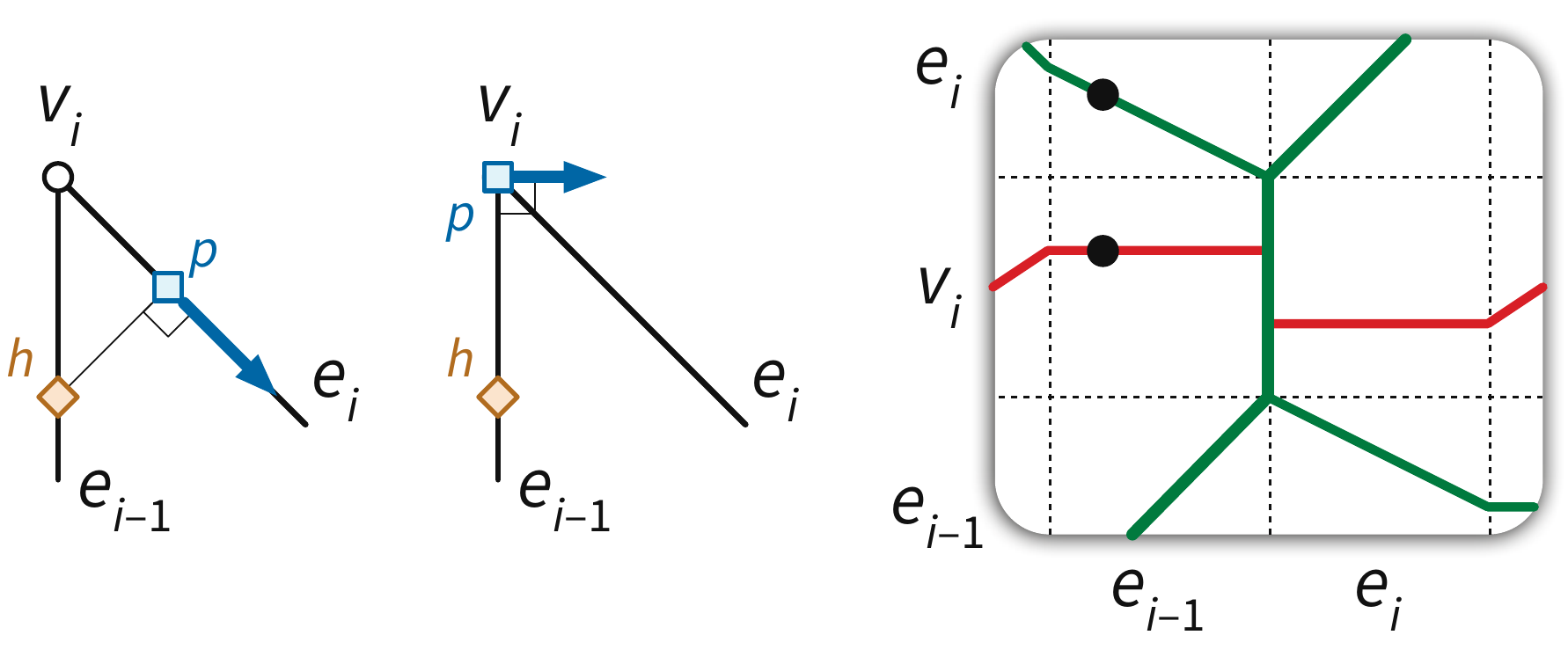}
\caption{Stable and unstable configurations near an acute vertex angle.}
\label{F:degen1}
\end{figure}

A \EMPH{type-2 degeneracy} is caused by a more specific configuration.
Let $e_i$ be an edge of~$P$, and let $\ell$ be the line perpendicular to $e_i$ through $v_i$ (or, symmetrically, through $v_{i+1}$). 
Let $v_j$ be another vertex of $P$ which lies on $\ell$. 
The configuration $(x, y)$ with $P(x) = v_j$ and $\pi(y) = v_i$ is degenerate if:
\begin {itemize} [noitemsep]
  \item $v_{i-1}$ and $v_j$ lie in the same open halfspace of the supporting line of $e_i$; \textbf{and}
  \item $v_{j-1}$ and $v_{j+1}$ lie in the same open halfspace of $\ell$.
\end {itemize}
A type-2 degeneracy corresponds to a vertex (pivot configuration) of degree 4 or 0 in the attraction diagram. We distinguish these further as \emph{type-2a} and \emph{type-2b}.
Refer to \cref{F:degen2}.

\begin{figure}[ht]
\centering
\includegraphics[scale=0.4]{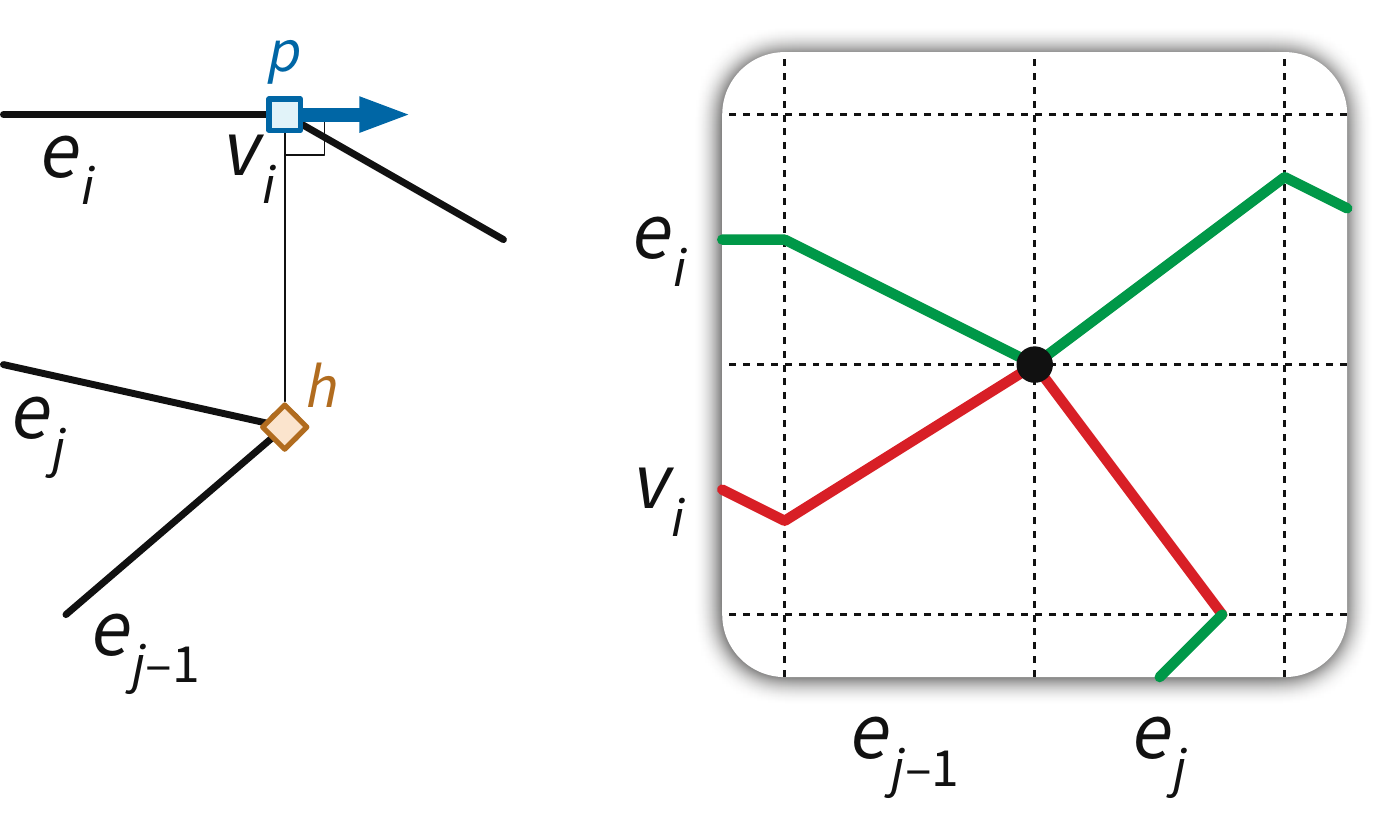}
\hfil
\includegraphics[scale=0.4]{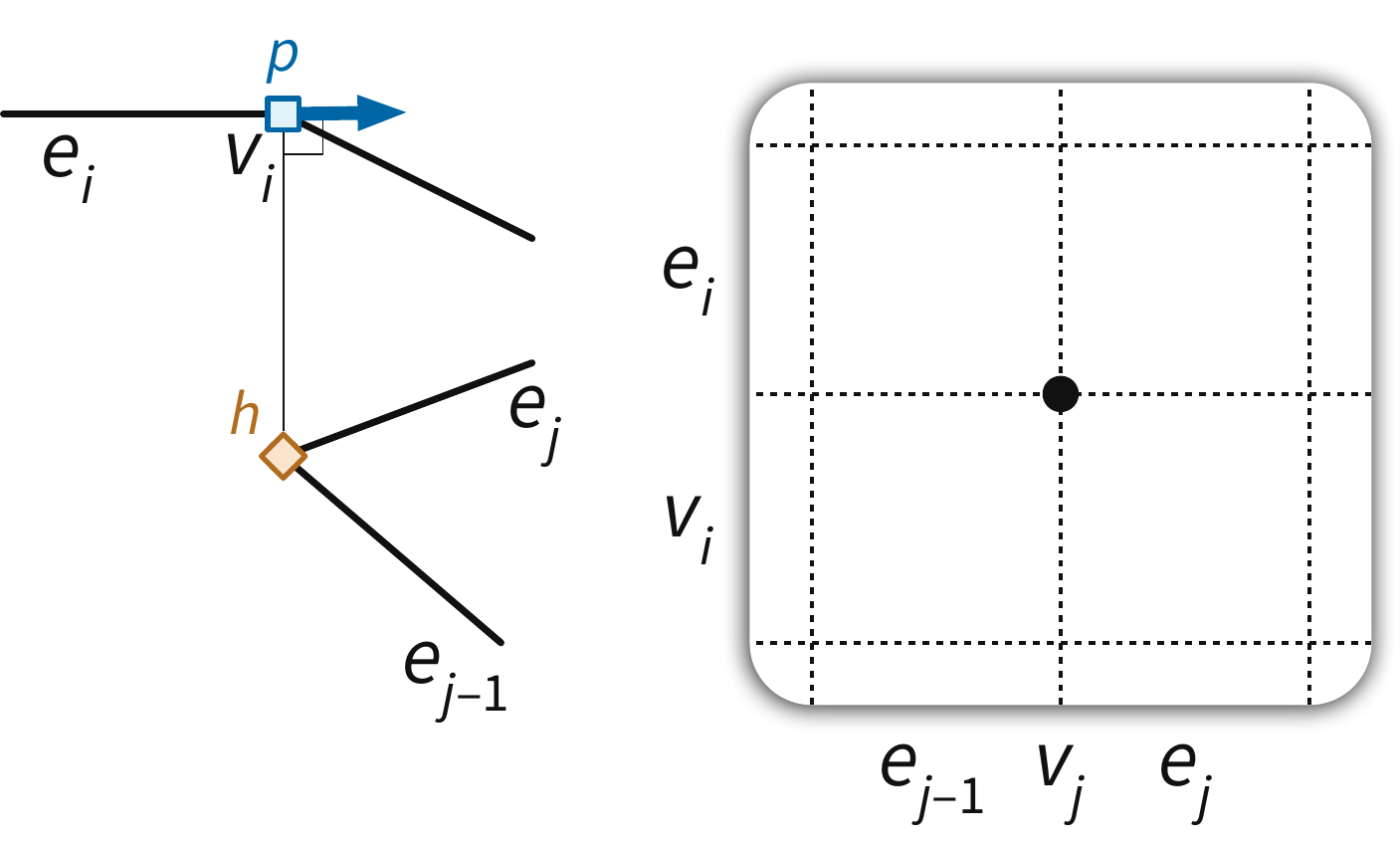}
\caption{Type-2a and type-2b degenerate pivot configurations.}
\label{F:degen2}
\end{figure}

Finally, a \EMPH{type-3 degeneracy} is essentially a limit of both of the previous types of degeneracies.  Let $e_i$ be an edge of $P$, let $\ell$ be the line perpendicular to $e_i$ through $v_i$, and let $e_j$ be another edge of $P$ which lies on $\ell$. 
The configuration $(x, y)$ with $P(x) \in e_j$ and $\pi(y) = v_i$ is degenerate if vertices  $v_{i-1}$ and $v_j$ lie in the same open halfspace of the supporting line of $e_i$.
When this degeneracy occurs, pivot configurations are not discrete, because the point $P(x) \in e_j$ can be chosen arbitrarily.  Moreover, the vertex-vertex configurations $(v_j, v_i)$ and $(v_{j-1}, v_i)$ have odd degree in the attraction diagram.  
A type-3 degeneracy can be connected to (2 or more) other critical curves, or be isolated. We distinguish these further as \emph{type-3a} and \emph{type-3b}.
See \cref{F:degen3}.

\begin{figure}[ht]
\centering
\raisebox{-0.5\height}{\includegraphics[scale=0.4]{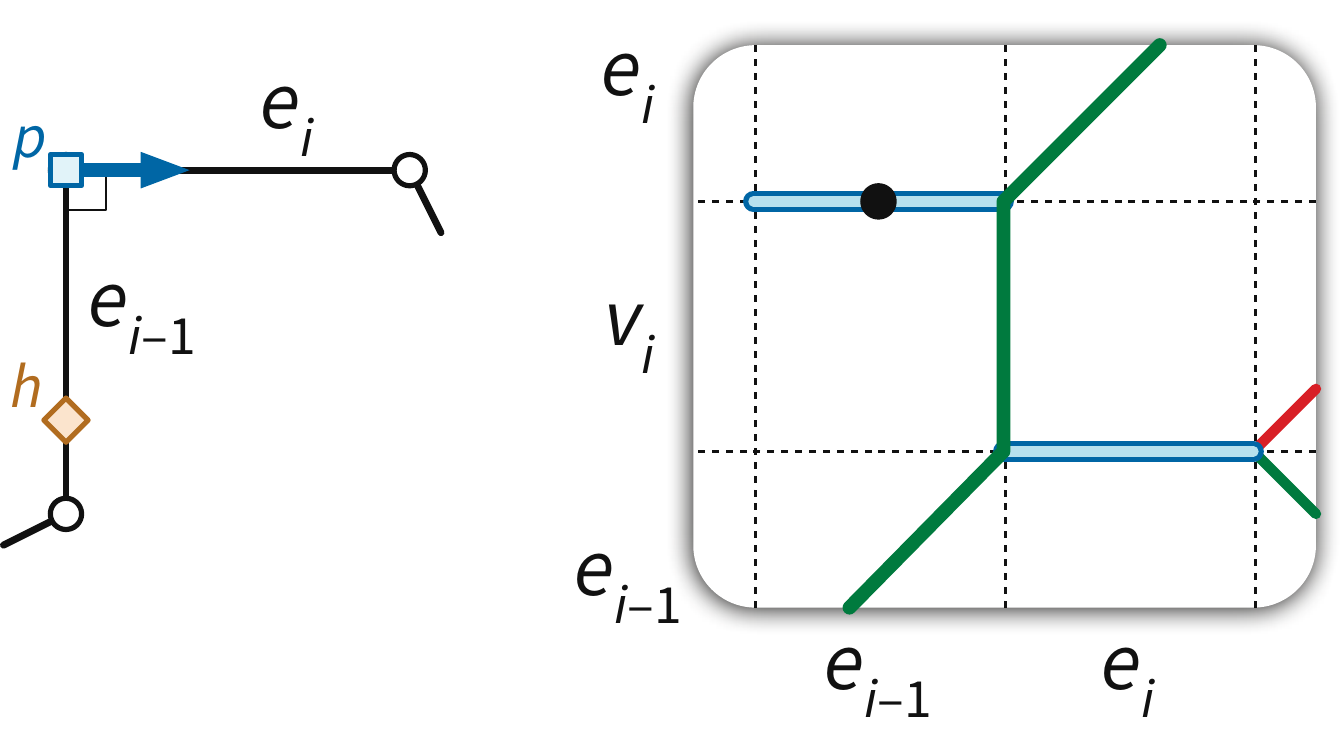}}
\hfil
\raisebox{-0.5\height}{\includegraphics[scale=0.4]{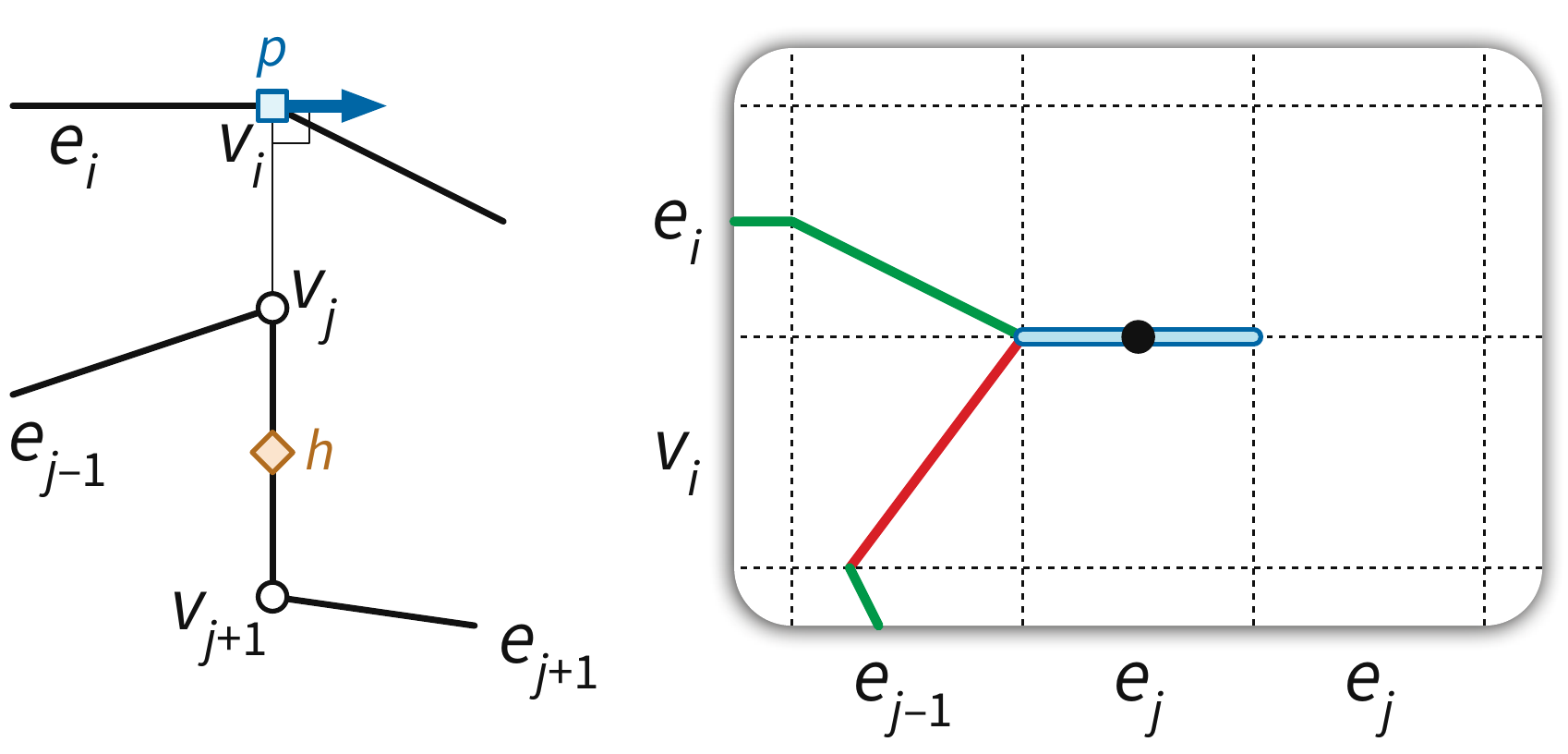}}
\caption{Type-3 degenerate pivot configurations on incident and non-incident edge pairs.}
\label{F:degen3}
\end{figure}

In \cref {sec:genericpolygons} we first consider polygonal tracks which do not have any degeneracies of these three types.  To simplify exposition, we forbid degeneracies by assuming that no vertex angle in $P$ is acute and that no three vertices of $P$ define a right angle.  In \cref {sec:degeneratepolygons} we lift these assumptions by \emph{chamfering} the polygon, cutting off a small triangle at each vertex.

\subsection{Catching puppies on generic obtuse polygons}
\label{sec:genericpolygons}

Generic obtuse polygonal tracks behave almost identically to smooth tracks, once we properly define the attraction diagram and dual attraction diagram.

\begin{lemma}\label{lem:nodedegeracy}
Let $P$ be a simple polygon with no acute vertex angles, in which no three vertices define a right angle.  The attraction diagram of $P$ is the union of disjoint simple critical cycles.
\end{lemma}

\begin{proof}
Each edge-edge cell $e_i\times e_j$ contains at most one section of stable critical configurations $(x,y)$ (\cref {F:edge-edge}).  For each such configuration, the points $\pi(y)\in e_i$ and $P(x)\in e_j$ are connected by a line perpendicular to $e_i$.  Because no three vertices of $P$ define a right angle, these points cannot both be vertices of $P$; thus, any critical path inside the cell $e_i\times e_j$ avoids the corners of that cell. 

Each vertex-edge cell $v_i\times e_j$ contains at most one section of a stable and one section of an unstable path (\cref {F:vertex-edge}). Again, because no three vertices of $P$ define a right angle, these paths avoid the corners of the cell $v_i\times e_j$.

For every pivot configuration $(x,y)$, the puppy $\pi(y)$ lies at a vertex $v_i$, the  puppy's direction $\theta(y)$ is parallel to either $e_i$ (or $e_{i+1}$), and because no three vertices of $P$ form a right angle, the human $P(x)$ lies in the interior of some edge $e_j$.  Since we avoid degenerate pivot configurations, each pivot configuration is a shared endpoint of an unstable critical path in cell $v_i\times e_j$ and a stable critical path in cell $e_i\times e_j$ (or $e_{i-1}\times e_j$).

It now follows that the set of unstable critical configurations is the union of monotone paths whose endpoints are pivot configurations.  Similarly, the set of stable critical configurations is also the union of monotone paths whose endpoints are pivot configurations.  Each unstable critical path lies in a single vertex strip.

Because every vertex angle in $P$ is obtuse, every configuration $(x,y)$ where the human $P(x)$ lies on an edge $e_i$ and the puppy $\pi(y)$ lies on the previous edge $e_{i-1}$ is either stable of final.  In particular, the only non-pivot critical configurations $(x,y)$ with $\pi(y) = v_i$ are the final configurations with $P(x) = v_i$.  Thus, the main diagonal is disjoint from all other critical cycles; in fact, no other critical cycle intersects any grid cell that touches the main diagonal.

\begin{figure}[ht]
\centering
\includegraphics[scale=0.4]{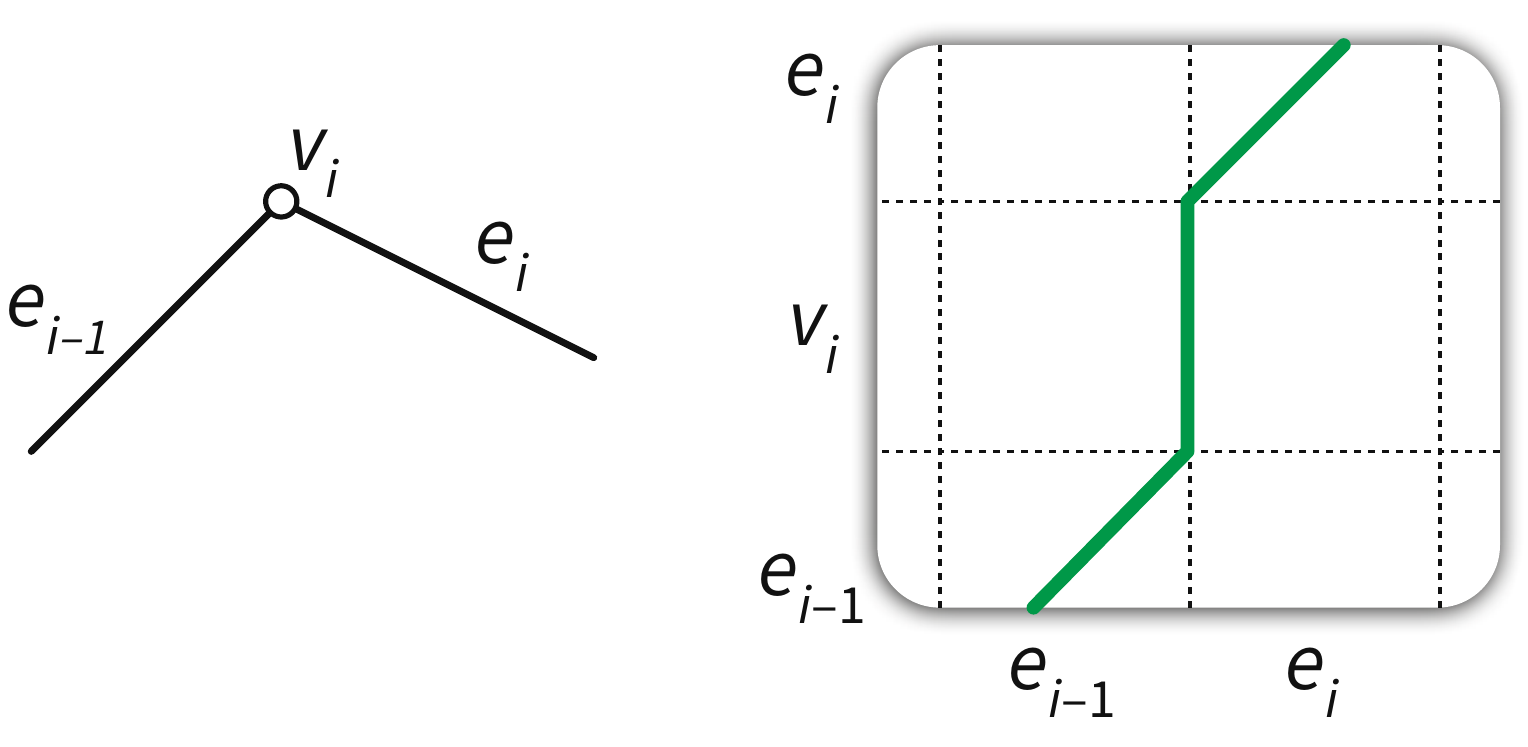}
\caption{Near the main diagonal.}
\label{F:main-diagonal}
\end{figure}

This completes the classification of all critical configurations.  We conclude that the attraction diagram consists of the (simple, closed) main diagonal and possibly other simple closed  curves composed of  stable and unstable critical paths meeting at pivot configurations.  All these critical cycles are disjoint.
\end{proof}

\begin{lemma}
Let $P$ be a simple polygon with no acute vertex angles, in which no three vertices define a right angle.  If the attraction diagram of $P$ has exactly two essential critical cycles, then the human can catch the puppy on $P$, starting from any initial configuration.
\end{lemma}

The remainder of the proof is essentially unchanged from the smooth case.  For any configuration $(x,y)$, let $T(y)$ denote the directed “tangent” line through $\pi(y)$ in direction $\theta(y)$, and let $L(x,y)$ denote the signed distance from $P(x)$ to $T(y)$, signed positively if $P(x)$ lies to the left of $T(y)$ and negatively if $P(x)$ lies to the right of $T(y)$.  The \emph{dual attraction diagram} of $P$ consists of all points $(y, L(x,y)) \in S^1\times \Real$ where $(x,y)$ is a critical configuration.  As in the smooth case, the map $(x,y) \mapsto (y, L(x,y))$ is a homeomorphism from the critical cycles in the attraction diagram to the curves in the dual attraction diagram; moreover, this map preserves the contractibility of each critical cycle.  

\begin{lemma}
\label{L:polygon-simple-two}
Let $P$ be a simple polygon with no acute vertex angles, in which no three vertices define a right angle.  The attraction diagram of $P$ contains exactly two essential critical cycles.
\end{lemma}

\begin{theorem}
\label{thm:genericpolygons}
Let $P$ be a simple polygon with no acute vertex angles, in which no three vertices define a right angle.  The human can catch the puppy on $P$, starting from any initial configuration.
\end{theorem}

\subsection{Chamfering}
\label{sec:chamfering}

  We now extend our analysis to arbitrary simple polygons. We define a {\em chamfering} operation which transforms a polygon $P$ into a new polygon $\cham{P}$. 
  First we show that $\cham{P}$ has no more degeneracies of type 1, 2a, or 3a. 
  The polygon $\cham{P}$ may still have degenerate pivot configurations of type~2b and type~3b; 
  however, since these correspond to isolated forward or backward pivot configurations in the attraction diagram, they do not impact the existence of a strategy to catch the puppy on $\cham{P}$. The puppy will just move over them as if they were normal forward or backward configurations.
  Thereafter we show that such a strategy can be correctly translated back to a strategy on~$P$.

  Let $P$ be an arbitrary simple polygon, and let $\e > 0$ be smaller than half of any distance between two non-incident features of $P$.
  Then the {\em $\e$-chamfered} polygon $\cham{P}$ is another simple polygon with twice as many vertices as $P$, defined as follows.
  Refer to \cref {fig:poly-chamfering}.
  For each vertex $v_i$ of $P$, we create two new vertices $v_i'$ and $v_i''$, where $v_i'$ is placed on $e_{i-1}$ at distance $\e$ from $v_i$, and $v_i''$ is placed on $e_i$ at distance $\e$ from $v_i$.
  Edge $e_i'$ in $\cham{P}$ connects $v_i''$ to $v_{i+1}'$, and a new {\em short edge} $s_i$ connects $v_i'$ to $v_i''$.

\begin{figure}[ht]
\centering
\includegraphics[width=\textwidth]{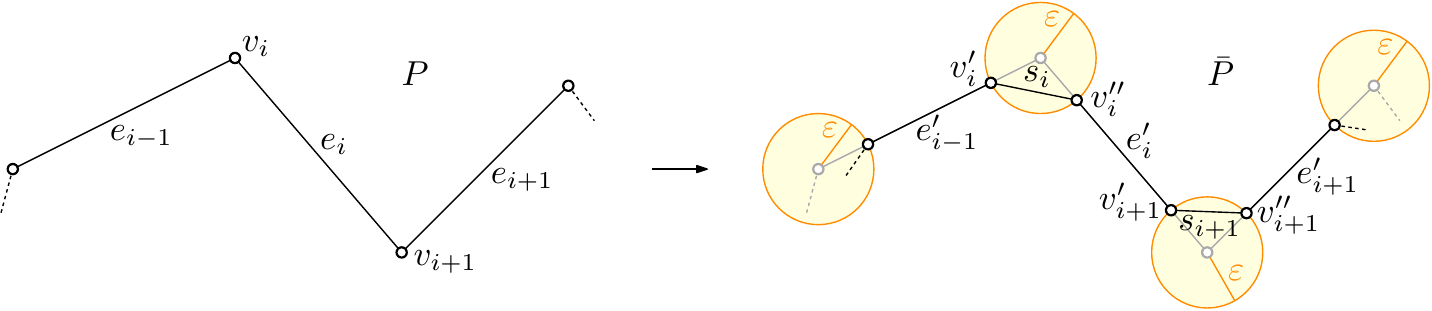}
\caption{The chamfering operation.}
\label{fig:poly-chamfering}
\end{figure}
    
  The chamfering operation alters the local structure of the attraction diagram near every vertex. 
  The idea is that for non-degenerate configurations, the change will not influence the behaviour of the puppy in such configurations, and as such will not influence the existence of any catching strategies. 
  However, at degenerate configurations, the change in the structure is significant. 
  We will argue in \cref{sec:degeneratepolygons} that the changes are such that every strategy  in the chamfered polygon translates to a strategy in the original polygon.
  
\begin{figure*}
\centering
\includegraphics[scale=0.4]{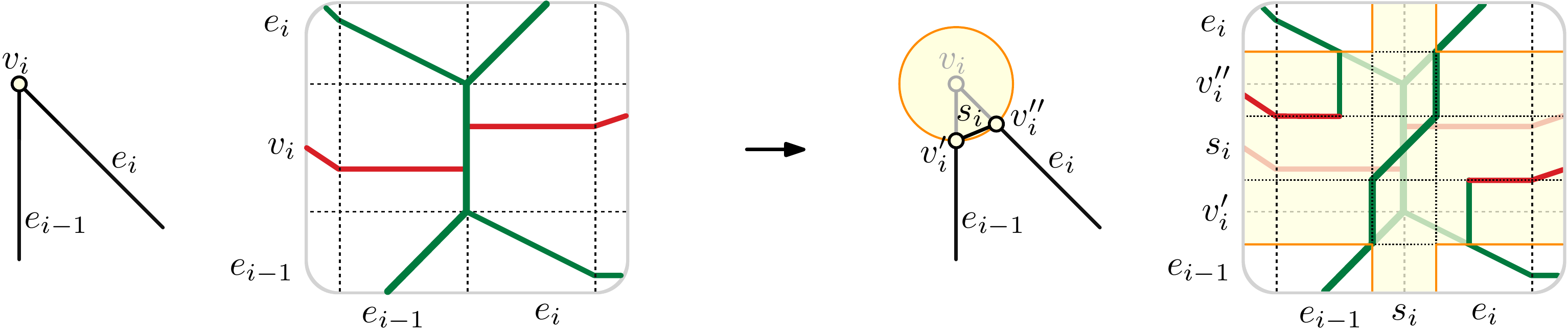}
\\{\bf type 1}\\[1em]
\includegraphics[scale=0.4]{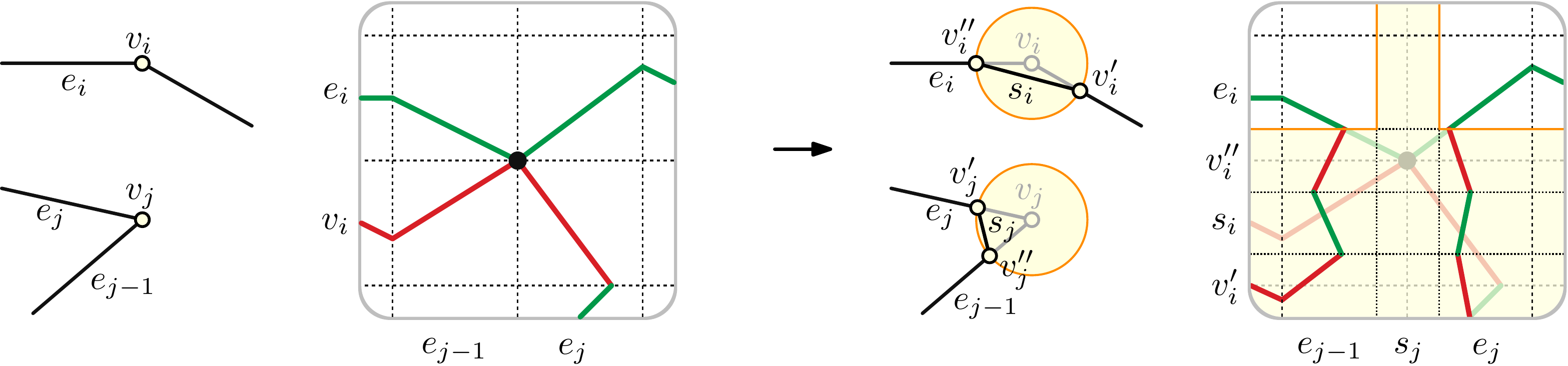}
\\{\bf type 2a}\\[1em]
\includegraphics[scale=0.4]{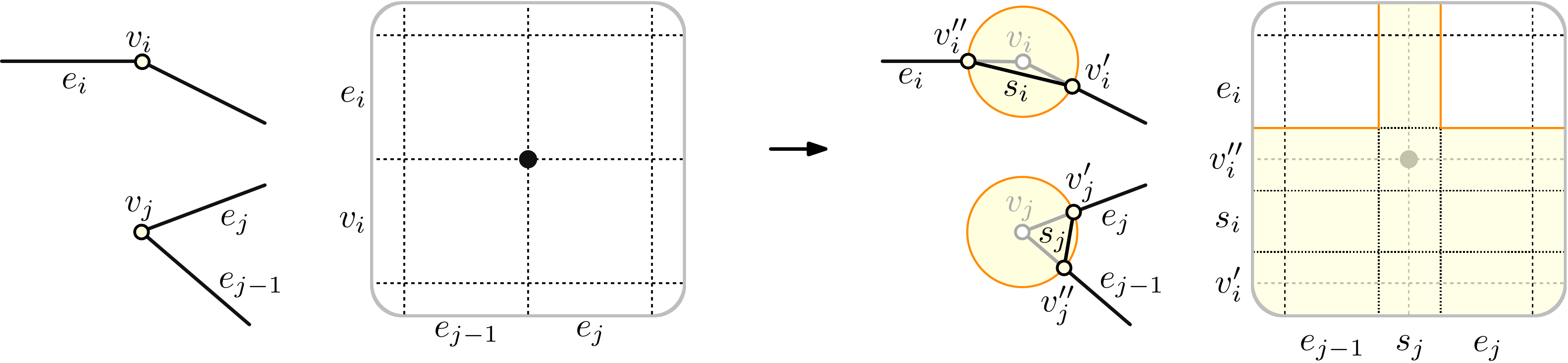}
\\{\bf type 2b}\\[1em]
\includegraphics[scale=0.4]{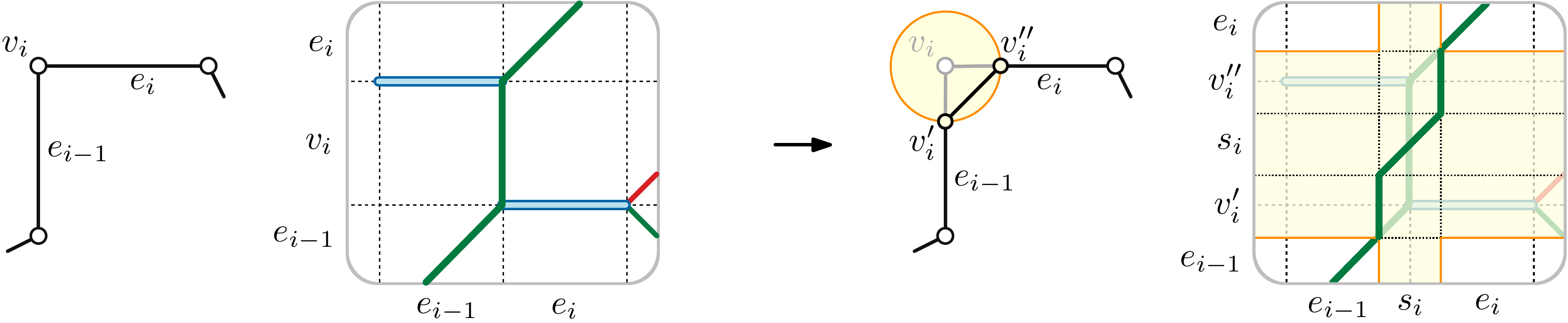}
\\{\bf type 3a}\\[1em]
\includegraphics[scale=0.4]{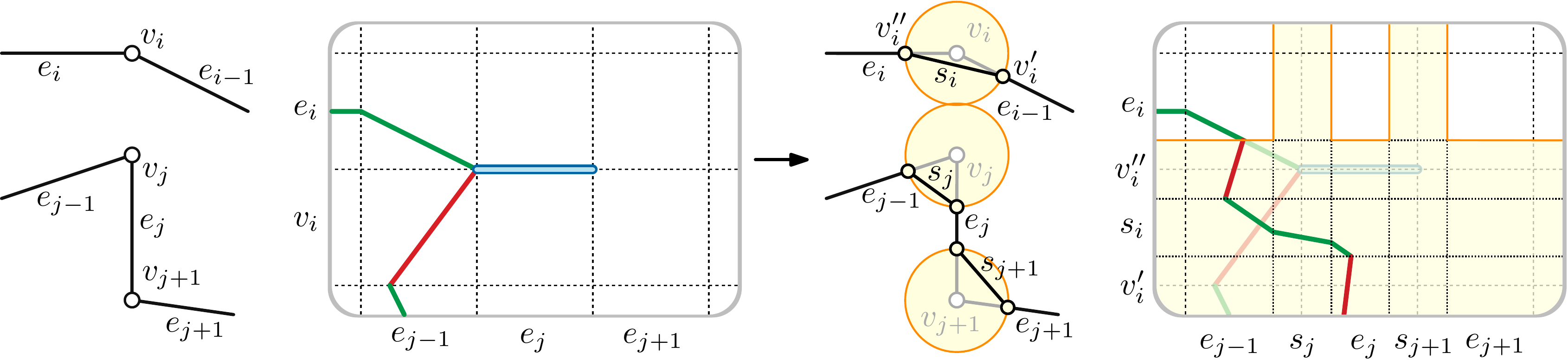}
\\{\bf type 3b}\\[1em]
\caption{Effect of the chamfering operation on the attraction diagram near degenerate pivot configurations. The size of $\e$ is exaggerated; the figures show the combinatorial structure of the chamfered diagram for a much smaller value of $\e$. Only the effect of chamfering vertices relevant for the degeneracy is shown.}
\label{F:degen-cham}
\end{figure*}

  Here we review again the different types of degenerate pivot configurations, and how the chamfering operation affects the structure of the attraction diagram in each case.
  Refer to \cref {F:degen-cham}.

  \begin {itemize}
    \item Near type~1 degeneracies, the higher-degree vertices on the main diagonal disappear. Instead, two separate critical curves almost touch the main diagonal: one from above and one from below.
    \item Near type~2a degeneracies, the degree-4 vertex disappears. Instead, the two incident critical curves coming from the left are connected, and the two incident curves coming from the right are connected.
    \item Near type-2b degeneracies, the isolated pivot vertex simply disappears.
    \item Near type-3 degeneracies, the degenerate pivot ``vertex'' disappears. Any connected critical curve is locally rerouted away from the degenerate location.
  \end {itemize}

\subsection{Catching puppies on arbitrary simple polygons}
\label{sec:degeneratepolygons}

  \begin {lemma}
  \label {lem:cham}
    Let $P$ be an abitrary simple polygon. 
    There exists an \e such that the \e-chamfered polygon~$\cham{P}$ has no degenerate pivot configurations of type 1, type~2a, or type~3a.
  \end {lemma}

  \begin {proof}
     First, note that $P'$ has no type 1 degeneracies: we replace each vertex $v_i$ with angle $\alpha_i$ by two new vertices $v_i'$ and $v_i''$ with angles $\alpha_i' = \alpha_i'' = \pi - \frac12(\pi - \alpha_i) = \frac12\pi + \frac12\alpha_i > \frac12\pi$.
     
     Next, we argue about type 2 degeneracies, which may occur for some values of $\e$.
     We argue that each potential type 2a degeneracy only occurs for a speficic value of $\e$; since there are finitely many potential degeneracies the lemma then follows.
     
     Note that, as we vary $\e$, all vertices of $P'$ move linearly and with equal speed.
     Suppose such a configuration is {\em not} unique for a specific choice of $\e$; that is, sustained for all values of $\e$. Then two vertices involved in the degenerate configuration must move in the same direction; that is, two edges of $P$, say $e_i$ and $e_j$, must be parallel.
     There are two configurations in $P'$ which could potentially give rise to a type 2 degeneracies.
     We argue that then, in fact, it cannot satisfy all requirements of a type 2 degeneracy.
     \begin {itemize} [noitemsep]
       \item An edge $e_i'$ has endpoint $v_i'$ (or symmetrically, $v_{i-1}''$) such that the line $\ell$ through $v_i'$ and perpendicular to $e_i'$ contains another vertex $v_j'$ (or $v_{j-1}'')$.  
       Refer to \cref {F:type2b-original}.
       Then, as $v_i'$ move along $e_i'$, $\ell$ move at the same speed as $v_i'$, so $v_j'$ moves in the same direction at the same speed along $e_{j-1}'$. So $e_{j-1}'$ is parallel to $e_i'$. 
       But now consider $v_j''$. Since $v_j$ lies on the supporting line of $e_{j-1}'$ at distance $\e$ from $v_j'$, and $v_j''$ lies at distance $\e$ from $v_j$, we conclude that $v_{j-1}''$ and $v_j''$ lie on the opposite side of $\ell$; thus, this is not a type 2 degeneracy.

\begin{figure}[ht]
\centering
\includegraphics{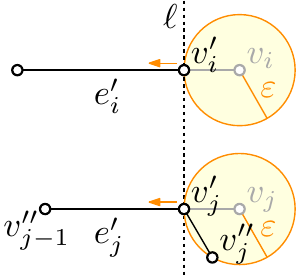}
\caption{Potential new degenerate pivot configurations based on a (shortened) original edge $e_i'$. For $\e$ small enough, there can be no degeneracy.}
\label{F:type2b-original}
\end{figure}

       \item A short edge $s_i$ of $P'$ has an endpoint $v_i'$ (or symmetrically, $v_{i}''$) such that the line $\ell$ through $v_i'$ and perpendicular to $s_i$ contains another vertex $v_j'$ (or $v_{j-1};'')$. 
       Refer to \cref {F:type2b-short}.
       Then, as $\e$ increases, $v_i'$ move along $e_i'$, $\ell$ moves at a slower speed; hence, we cannot conclude anything about the orientation of $e_j$. However, in this case,  note that $e_i'$ and $s_i$ lie on opposite sides of $\ell$; therefore, $e_j'$ and thus $v_{j-1}''$ lies on the opposite side of $\ell$. It is possible that $v_j''$ lies on the same side, in which case we have a degenerate pivot configuration of type 2b (\cref {F:type2b-short} (left)), or that $v_j''$ lies on $\ell$, in which case we have a degenerate pivot configuration of type 3b (\cref {F:type2b-short} (middle)). If $v_j''$ lies on the opposite side of $\ell$, there is no degeneracy (\cref {F:type2b-short} (right)).
     \end {itemize}
  \end {proof}

\begin{figure}[ht]
\centering
\includegraphics{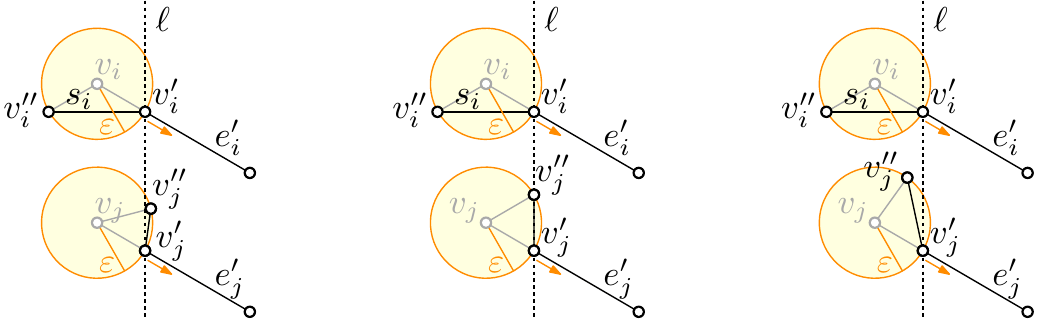}
\caption{Potential new degenerate pivot configurations based on a short edge $s_i$. For any $\e$ we may still have a new degeneracy of type 2b (left), 3b (middle), or no degeneracy (right).}
\label{F:type2b-short}
\end{figure}
   
\begin{figure*}
\centering
\begin{tabular}{c@{\qquad}c}
\includegraphics[scale=0.5,page=1]{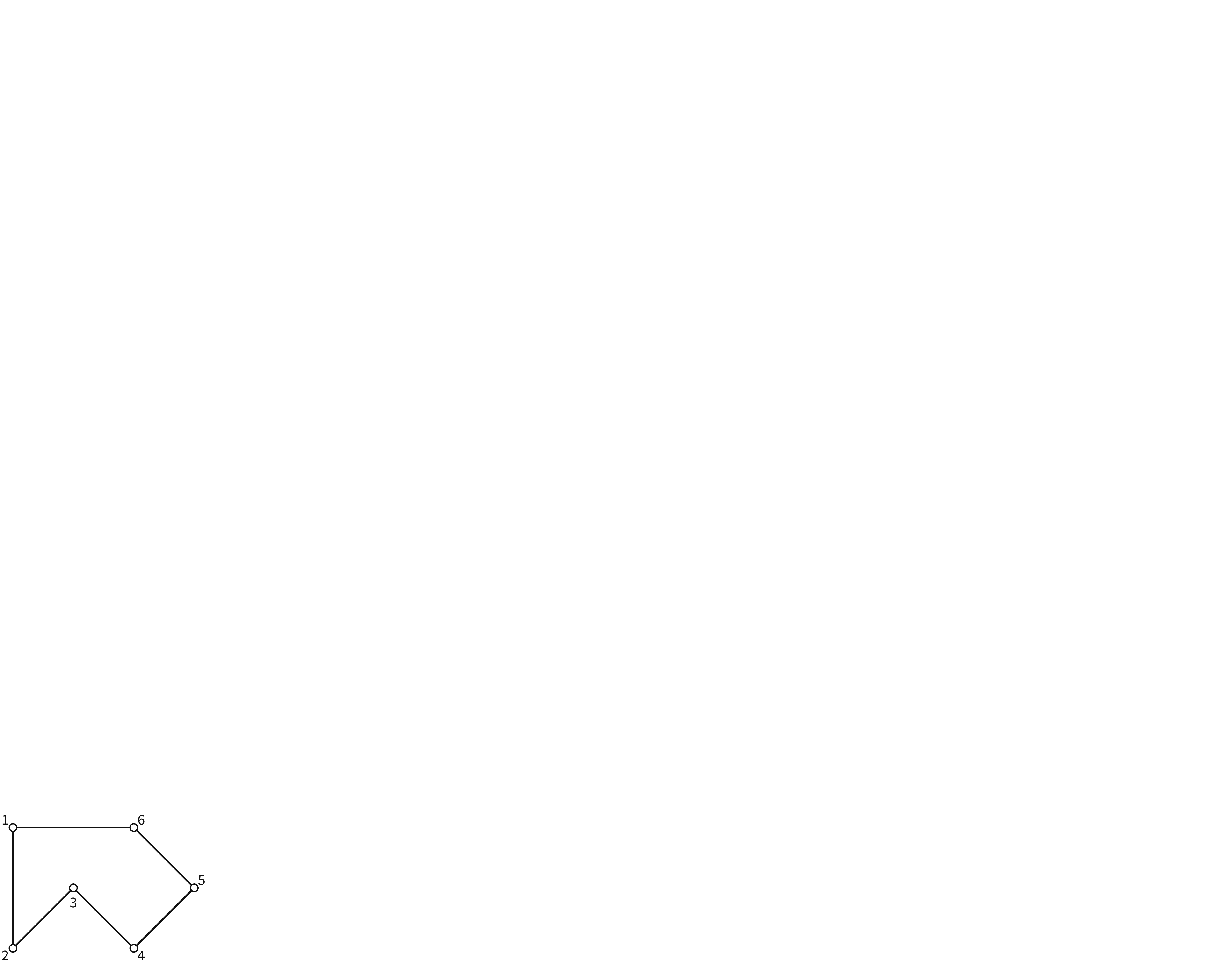} &
\includegraphics[scale=0.5,page=1]{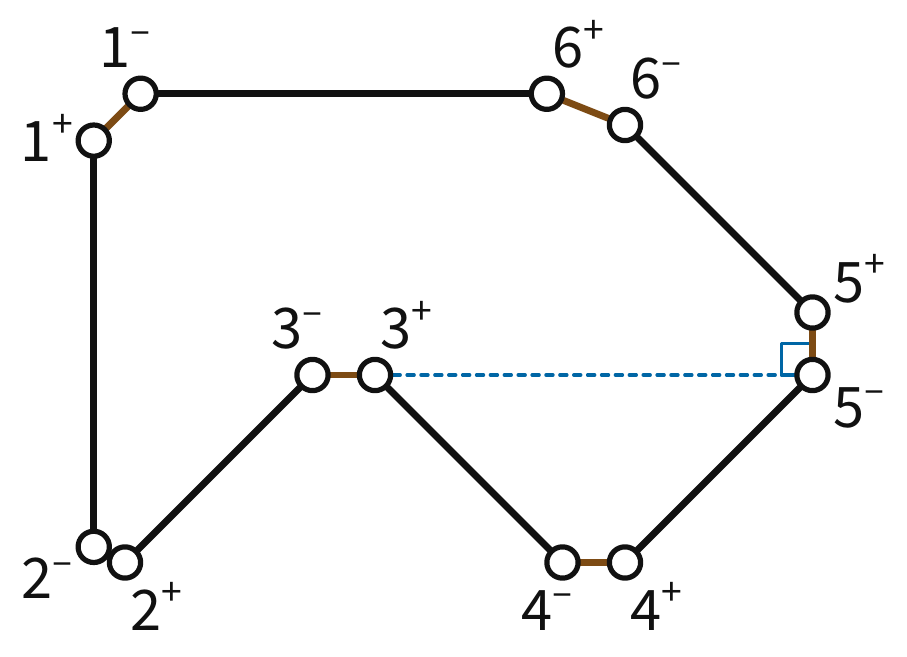} \\
\includegraphics[scale=0.375,page=2]{Fig/stonks} &
\includegraphics[scale=0.375,page=2]{Fig/stonks-chamfered}
\end{tabular}
\caption{The attraction diagram of a degenerate polygon, before and after chamfering.  All existing degeneracies disappeared in the chamfered polygon, which does have one new but harmless degeneracy.}
\label{fig:poly-puppy}
\end{figure*}

    Let $P$ be an arbitrary simple polygon and $\cham{P}$ a chamfered copy without degeneracies.
    We say a parameter value $x$ is {\em verty} whenever $P(x)$ is  within distance $\e$ from a vertex.
    We say a parameter value $x$ is {\em edgy} whenever $P(x)$ is not within distance $\e$ from a vertex.
    We reparameterize $\cham{P}$ such that $P(x) = \cham{P}(x)$ whenever $x$ is edgy; the parameterization of $\cham{P}$ is uniformly scaled for verty parameters.
    We say a configuration $(x, y)$ is edgy when $x$ and $y$ are both edgy.

    We say a path~$\sigma$ is a \emph{valid} path in the attraction diagram
    if it describes a human and puppy behavior that obeys the rules 
    imposed on the puppy and the human, as explained in \cref{S:intro}.

\begin {lemma}
  \label {lem:equivalent}
    Assuming $\e$ is sufficiently small.
    For any valid path $\sigma$ between two edgy  configurations $(x_1, y_1)$ and $(x_2, y_2)$ in the attraction diagram of $\cham{P}$, there is an equivalent path between $(x_1, y_1)$ and $(x_2, y_2)$ in the attraction diagram of $P$.
\end {lemma}

\begin{proof}
    The intuition is that the behavior of the puppy around vertices is the same as in the original polygon.
    Note that when the human reaches a vertex in~$P$ 
    then it can happen that the puppy suddenly moves infinitely fast until it reaches a stable position.
    However, if the human is \emph{leaving} a vertex then the puppy may move, but only at a finite speed.
    (At least for some short distance.)
    Thus if the puppy is at an edgy position, we can always move the human
    such that we get into an edgy configuration, given that \e is small enough.
    Recall that there is a one to one correspondence between edgy configurations of $P$ and $\cham{P}$.
    Also once the puppy is in an edgy configuration, we can also move
    the human into an edgy configuration, assuming we have chosen \e small enough.
    Thus, we only care about the situation that the puppy is in a verty configuration. We need to make sure that the same way the puppy can leave the verty position in \cham{P} can be simulated in $P$. 
    
    \begin{figure}
        \centering
        \includegraphics[page = 2, scale = 0.8]{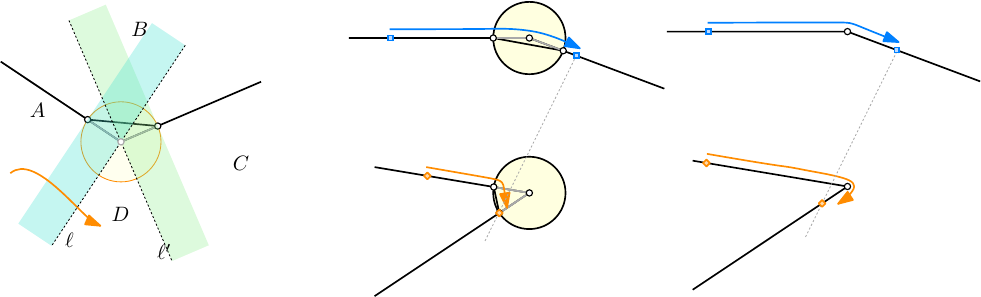}

        \caption{Left: When the puppy is at the vertex~$v_i$, then its behavior is determined on the region that the human lies. 
        We have four regions, called~$A,B,C,D$.
        Right: Once, the human gets into a verty configuration, the puppy makes an infinitely fast jump forward.
        This can be simulated in $P$, as the lower vertex stems from a type~2a degeneracy.}
        \label{fig:strategy-preservance}
    \end{figure}
    
    We first handle the situation that the puppy traverses
    over a generic vertex $v_i$.
    To be precise, let~$v_i$ be a vertex of~$P$ 
    and~$e_{i-1}$ and~$e_{i}$ its incident edges. 
    Let $\ell$ and $\ell'$ be the lines through $v_i$ orthogonal to $e_{i-1}$ and $e_i$ respectively.
    We say that $v_i$ is generic if there is no  vertex that
    lies on either $\ell$ or $\ell'$.
    We denote the infinite strip bounded by $\ell$ and $v_i'$ by the letter~$S$.
    Similarly, we denote the infinite strip bounded by $\ell$ and $v_i'$ by the symbol~$S'$.
    As no vertex of $P$ lies on one of the two lines~$\ell, \ell'$ then, we can choose \e small enough such that also there is no vertex of \cham{P} in one of the strips~$S,S'$, other than $v_i'$ and $v_i''$.
    
    Thus whenever the human crosses over one of the strips $S$
    or $S'$, the human is on an edge of~$\cham{P}$. 
    Thus, when the puppy will leave the verty position defined by the segment $s_i$, then the human will be in an edgy configuration, as is the puppy.
    Thus we can simulate the same valid path moving the puppy out of the verty configuration at $s_i$.
    
    We are now handling the case that the puppy moves towards a degenerate configuration. 
    Recall all degenerate positions as listed in \cref{F:degen-cham}.
    Interestingly, while those configurations are degenerate for the attraction diagram the behavior of the puppy and the human is easily understood.
    For instance in a type~1 degeneracy the human walks 
    towards the acute angle and the puppy will move to the same position.
    Note that type~2b configurations are not changing the behavior of the puppy at all, as the isolated pivot point will never be reached in the attraction diagram.
    Similarly, type~3a and type~3b degeneracies are not altering
    the puppy behavior in an interesting way.
    Given a movement of the human in~\cham{P} it is straightforward how the movement of the human should be in~$P$. 
    See to the right of \cref{fig:strategy-preservance}, for the most difficult example of a type~2a degeneracy.
\end{proof}

  We are now ready to prove our main result.

\begin{theorem}
Let $P$ be a simple polygon.  The human can catch the puppy on $P$, starting from any initial configuration.
\end{theorem}

\begin {proof}
  By \cref {lem:cham}, there is a value $\e$ such that the $\e$-chamfered polygon $\cham{P}$ has no degeneracies of type 1 or type~2a or type~3. There may still be some type~2b degeneracies, but they result in isolated pivots in the 
  attraction diagram, which can be safely ignored.

  Consider an arbitrary starting configuration on $P$. If the starting configuration is not stable, we let $p$ move until it is. 
  If the resulting configuration is not edgy, we walk $h$ along~$P$ until we reach an edge configuration $(x, y)$.
  (This must be possible, except if the puppy stays on a vertex 
  for the entire time, in that case, we can catch the puppy trivially, by going to that vertex.)
  
  By \cref {thm:genericpolygons}, there exists a strategy for $h$ to catch $p$ on $\cham{P}$.
  If the end configuration of this strategy is not edgy, note that we may now simply move $h$ and $p$ together to an edgy final configuration $(f, f)$.
  By \cref {lem:equivalent}, there is an equivalent strategy to reach $(f, f)$ from $(x, y)$ on $P$. Combined with the initial path to $(x, y)$ this gives us a path from an arbitrary starting configuration to a final configuration on~$P$.
\end {proof}

\section{Further questions}

For simple curves, we have only proved that a catching strategy exists.  At least for polygonal tracks, it is straightforward to compute such a strategy in $O(n^2)$ time by searching the attraction diagram.  In fact, we can compute a strategy that minimizes the total distance traveled by either the human or the puppy in $O(n^2)$ time, using fast algorithms for shortest paths in toroidal graphs \cite{hkrs-fspap-97,ght-stgbg-84}.  Unfortunately, this quadratic bound is tight in the worst case if the output strategy must be represented as an explicit path through the attraction diagram.  We conjecture that an optimal strategy can be described in only $O(n)$ space by listing only the human’s initial direction and the sequence of points where the human reverses direction.  On the other hand, an algorithm to compute such an optimal strategy in subquadratic time seems unlikely.

If the track is a \emph{smooth curve} of length $\ell$ whose attraction diagram has $k$ pivot configurations, a trivial upper bound on the distance the human must walk to catch the puppy is $\ell\cdot k/2$.  In any optimal strategy, the human walks straight to the point on the curve corresponding to a pivot located at one of the two endpoints of the current ``stable sub-curve'' of a critical curve (walking less than $\ell$).  Then the configuration moves to another stable sub-curve, and so on, never visiting the same stable sub-curve twice.  Our question is whether a better upper bound can be proved.

In fact, if minimizing distance is not a concern, we conjecture that \emph{no} reversals are necessary.  That is, on \emph{any} simple track, starting from \emph{any} configuration, we conjecture that the human can catch the puppy \emph{either} by walking only forward along the track \emph{or} by walking only backward along the track.  \cref{F:intro2} and its reflection show examples where each of these naïve strategies fails, but we have no examples where both fail.  (Our proof of Theorem \ref{Th:ortho} implies that the human can always catch the puppy on an \emph{orthogonal} polygon by walking \emph{at most once} around the track in some direction, depending on the starting configuration.)

More ambitiously, we conjecture that the following \emph{oblivious} strategy is always successful: walk twice around the track the track in one (arbitrary) direction, then walk twice around the track in the opposite direction.

Another interesting question is to what extent our result applies to self-intersecting curves in the plane, when we consider the two strands of the curve at an intersection point to be distinct.  It is easy to see that the human cannot catch the puppy on a curve that traverses a circle twice; see \cref{F:doubleLoop}.  Indeed, we know how to construct examples of bad curves with any rotation number \emph{except} $-1$ and $+1$.  We conjecture that \cref{L:simple-good}, and therefore our main result, extends to all non-simple tracks with rotation number $\pm 1$.  Similarly, are there interesting families of curves in $\Real^3$ there the human and puppy can always meet?

Finally, it is natural to consider similar pursuit-attraction problems in more general domains.  \cref{thm:pol} shows that the human can always catch the puppy in the interior of a simple polygon, by walking along the dual tree of any triangulation.  Can the human always catch the puppy in any planar straight-line graph?  Inside any polygon with holes?  

\bibliography{lib}

\end{document}